\documentclass[10pt,letter,journal]{IEEEtran}
\IEEEoverridecommandlockouts
\usepackage{amsmath,amscd,amssymb,amsgen,amsfonts,amsbsy}
\usepackage{mathrsfs}
\usepackage[utf8x]{inputenc}
\usepackage{color}
\usepackage{textcomp}
\usepackage{float}
\usepackage{latexsym,graphicx}
\usepackage{mathtools}
\usepackage{tikz}
\usetikzlibrary{automata,arrows,positioning,calc}
\usepackage{url}
\usepackage{cite}
\usepackage{algorithmic}
\usepackage[ruled,lined]{algorithm2e}
\usepackage{breqn}
\newcommand{\prob}[1]{\mathsf{Pr}\left( #1 \right)}

\newcommand{\remove}[1]{}

\newcommand{\comments}[1]{}

\newcommand{\qed}{\hfill $\square$}

\newcommand{\bm}[1]{\boldsymbol{#1}}

\newtheorem{lemma}{Lemma}

\newtheorem{theorem}{Theorem}
\newtheorem{remark}{Remark}

\newtheorem{proof}{Proof}
\newtheorem{definition}{Definition}
\makeatother

\title{Sequential Decision Making with Limited Observation Capability: Application to Wireless Networks}
\author{
 \begin{tabular}{ccc}
   Kesav Kaza,  Rahul Meshram, Varun Mehta  and S.~N.~Merchant \\
   Department of Electrical Engineering, \\
    IIT Bombay, Mumbai INDIA.  
  \end{tabular}
  \thanks{A preliminary version of this work was presented at IEEE WCNC 2018~\cite{Kaza2018wcnc}. The corresponding author for this paper is available at krk@ee.iitb.ac.in.}}
    
\begin{document}

\maketitle
\begin{abstract}
This work studies a generalized class of restless multi-armed bandits with hidden states and allow \textit{cumulative} feedback, as opposed to the conventional instantaneous feedback. 
We call them lazy restless bandits (LRB) as the events of decision-making are sparser than events of state transition. Hence, feedback after each decision event is the cumulative effect of the following state transition events.
 The states of arms are hidden from the decision-maker and rewards for actions are state dependent. The decision-maker needs to choose one arm in each decision interval, such that long term cumulative reward is maximized.

As the states are hidden, the decision-maker maintains and updates its belief about them.
It is shown that LRBs admit an optimal policy which has threshold structure in belief space.
The Whittle-index policy for solving LRB problem is analyzed; \textit{indexability} of LRBs is shown. 
Further, closed-form index expressions are provided for two sets of special cases; for more general cases, an algorithm for index computation is provided.  An extensive simulation study is presented; Whittle-index, modified Whittle-index and myopic policies are compared. Lagrangian relaxation of the problem provides an upper bound on the optimal value function; it is used to assess the degree of sub-optimality various policies. 

\end{abstract}


\section{INTRODUCTION}

\subsection{Motivation}

Wireless communication systems often operate in uncertain environments due to rapidly varying channel conditions and relative mobility of communicating nodes. Decision making under uncertainty occurs
in the problems of relay selection\cite{Li2011}, relay employment in wireless networks \cite{Kaza17}, opportunistic channel sensing and scheduling \cite{LiuZhao10},\cite{Wang14}, and downlink scheduling in heterogeneous networks\cite{Ouyang12}.

Let us consider the problem of relay selection in the following scenario.  Consider a wireless relay network with a source (S), destination (D) and a set of relays denoted by $(R_i),$ $ 1 \leq i \leq M-1.$ Suppose that the channels between source to relay and relay to destination operate at different frequencies. There 
are $M$ paths or links from source to destination that include 
the direct SD link and source-relay-destination links. Further, channel quality along each of these paths is time varying. The time is divided into intervals. The objective of a source is to use $N$ paths out of $M$ in each interval such that it maximizes the expected long term throughput. The source cannot observe the exact channel qualities along each path. This introduces significant difficulty in decision making (relay selection); it can be resolved by using a feedback mechanism. When a certain link is used for transmission, feedback is available at the end of the interval as ACK/NACK, which signify success or failure of the message transmission. No feedback is available from the unused paths. The source forms a belief about the channel qualities of the used paths based on the feedback. Using this information, the source selects $N$ paths in the given interval.  
	
The above application is an example of a sequential decision problem which can be described abstractly as follows. There is a decision maker that interacts with a system or environment through a given set of actions. The environment responds to each action differently by changing its state and also generating a reward. The decision maker can fully or partially observe the environment states and rewards. It's goal is to choose actions along the time line to maximize the expected cumulative reward. Clearly, to this end the decision making strategy must consider the effect of current action on future rewards.
Some sequential decision making scenarios exist where the system is composed of seemingly independent entities. Also, the statistical behavior of the system (eg. transition probability matrix, average rewards, etc.) is known to the decision maker. In such cases, the decision maker needs to plan which action it would choose when it observes a certain system state. This is called a planning problem, and it can be modeled using restless multi-armed bandits.

\subsection{Restless multi-armed bandits (RMAB)}
An RMAB has a set of independent arms. At each time step the decision maker plays a fixed number of arms. The states of arms evolve independently at each time step. The play of arms yields state dependent rewards. The arms which are not played yield no reward. The objective of the decision maker is to determine the optimal sequence of plays which maximizes long term discounted cumulative reward. RMAB was first proposed in \cite{Whittle88} as a generalization of the classical multi-armed bandit problem in \cite{Gittins74}. In an RMAB each arm is a Markov decision process (MDP) or partially observable Markov decision process (POMDP) depending on whether the states are fully or partially observable. These processes are coupled by the constraint that only a fixed number of them can be activated (played) in a given time step.

In conventional RMAB models, system state transition and decision making occur at discrete time instants uniformly spaced along the time line. The knowledge of the system state at these instants, provides information that is necessary for decision making. 

In this paper we consider a scenario where the information gathering of the decision maker is not at par with the variation of system state. The instants of decision making are sparse compared to the instants of system state transition.  The decision maker does not observe every state transition; instead, observation of the system takes place only when a decision needs to be made. This is due to the \textit{limited observation capability} of the decision maker. We refer to the information gathered by this form of observation as \textit{cumulative feedback}; it represents the cumulative effect of a series of state transitions. In the RMAB setting we can say that the bandit is \textit{lazy} in gathering information. This model allows multiple state transitions in one decision epoch. We call such restless bandits with cumulative feedback as \textit{lazy restless bandits.}

\subsection{Relay Selection as a Restless bandit problem}
The relay selection problem can be modeled as a RMAB problem. Each source-relay-destination link in a relay network corresponds to an arm.

We can model each source-relay-destination or source-destination link using a finite state Markov chain, where each state represents a certain channel quality. 
Since the channel qualities are not exactly observable by the source, the state of each link is not known. Thus each link can modeled as a POMDP. Further each link can be assumed to be independent of others.

While formulating the relay selection problem, two parameters play a key role - (1) decision interval, which is the time length for which each chosen relay is active, (2) slot length, which is the minimum time length over which the channel quality is assumed to remain constant. Choosing the decision interval to be equal to slot length might lead to significant signaling overhead and also increase delay. Hence, in our model we will assume that a decision interval consists of several slots. 

Markovian ON-OFF fading models have been used in literature for formulating such problems in order to account for temporal correlation of channel quality states during decision making \cite{LiuZhao10,Wang14,Ouyang12}. Although an ON-OFF model is a lossy representation of fading channels, it aids in taking decisions which are inherently of `threshold type'. For example, while employing a relay the source might require the end to end signal to noise ratio to cross a certain threshold in order to ensure quality of service to the end user. Such a requirement might be abstractly captured by a two state channel model. Further, in relay selection problem, the duration of using each relay (length of decision interval) is chosen such that the signaling overhead corresponding to the relay choice, is not too high. Hence, it might happen under fast fading conditions that there are intermediate channel variations during a single decision interval. This issue is also be addressed in our model. We believe the analysis of the relay selection problem using RMABs with two-state model would provide important insights which can be used for solving the multi-state model.

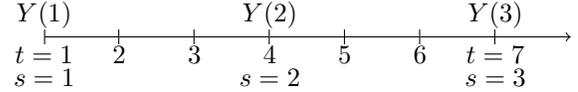
\begin{figure}[h]
  \begin{center}
  \begin{tabular}{c}  
\begin{tikzpicture} 
\draw [|-|](0,0)--(1,0); 
\node [below] at (0,0) {$t=1$}; 
\node [below] at (0,-0.3) {$s=1$}; 
\node [above] at (0,0) {$Y(1)$};
\draw [-|](1,0)--(2,0); 
\node [below] at (1,0) {$2$}; 
\draw [-|](2,0)--(3,0); 
\node [below] at (2,0) {$3$}; 
\draw [-|](3,0)--(4,0); 
\node [below] at (3,0) {$4$};
\node [below] at (3,-0.3) {$s=2$}; 
\node [above] at (3,0) {$Y(2)$};
\draw [-|](4,0)--(5,0); 
\node [below] at (4,0) {$5$}; 
\draw [-|](5,0)--(6,0); 
\node [below] at (5,0) {$6$}; 
\node [below] at (6,0) {$t=7$}; 
\node [below] at (6,-0.3) {$s=3$};
\node [above] at (6,0) {$Y(3)$}; 
\draw [->](6,0)--(7,0); 
\end{tikzpicture} 
\end{tabular}
\end{center} 
\vspace{-0.4cm}
\caption{Conventional model: Instants of state transition,
    observation and decision making are at $t.$ Cumulative feedback model: Instants of state transition are at $t$, observation and decision making are at $s.$} 
\label{fig:limfeemodel} 
\end{figure} 
\section{Literature Overview and Contributions} 
We now summarize the related literature. 
The RMAB problem was first introduced in \cite{Whittle88}, where the author studied a heuristic index policy which maps the state of each arm to a real number, and $N$ arms with highest indices are played at each time step. This   policy is now referred to as Whittle-index policy.  The key ideas involved in obtaining the indices are as follows. 1) Introduce a Lagrangian relaxation of the original coupled optimization problem. 2) This allows to decouple the original problem into sub-problems. Then, it is enough to solve the single armed bandit problem, where a Lagrangian variable is introduced as reward for not playing the arm. This can be interpreted as the subsidy for not playing the arm. 3)
It plays the role of an index.

	Motivated from a computational perspective, the linear programming approach for classical multi-armed bandit was developed in \cite{Chen86,Kallenberg86,Bertsimas-Nino-Mora96}, and was extended to RMAB in \cite{Bertsimas-Nino-Mora00}.  In \cite{Bertsimas-Nino-Mora96}, the authors  introduce  performance measure approach to model the classical multi-armed bandit problem. It is formulated as an linear program (LP) over an extended polymatroid constraint set. These constraints follow the conservation laws. Moreover, a priority index policy is derived using adaptive one pass greedy algorithm for LP. Similarly, the ideas in \cite{Bertsimas-Nino-Mora96} are extended to RMAB in \cite{Bertsimas-Nino-Mora00}.  Here, LP relaxation is introduced and priority index policies are developed via primal-dual heuristics. 

A generalization of MDPs/RMAB called weakly coupled Markov decision processes (WC-MDPs) was introduced in \cite{Hawkins03}. It consists of $M$ independent MDPs which are coupled through linking constraints. The objective is to maximize the long term total cumulative reward subject to the linking constraints. The Lagrangian relaxation for WC-MDP is introduced and is solved using linear programming algorithms and subgradient schemes. An extensive analysis of WC-MDP has been carried out in \cite{Adelman08}. In this, an LP based approximate dynamic programming (LP-ADP) approach to WC-MDP is introduced. 
An upper bound to the optimal value function of WC-MDP that is obtained through LP-ADP, is shown to be tighter than the Lagrangian relaxation bound. Moreover, it is shown that the gap between LP-ADP or Lagrangian relaxed value and optimal solution value is sub-linear in number of sub-problems.  Further, the numerical computations are performed by column generation simplex algorithm. 
Recently, a variation of WC-MDP called decomposable MDPs, have been studied in \cite{Bertsimas16}. An approximate solution approach based on a fluid linear optimization formulation is proposed. It is shown that this formulation provides a tighter bound on the optimal value function than the classical Lagrangian relaxation technique. This is illustrated via numerical examples for multi-armed bandit problems. 

A finite time horizon version of RMAB studied in \cite{Brown17} is motivated from applications such as dynamic assortment and applicant screening. Heuristics based on Lagrangian relaxation of the dynamic program, Whittle index based policy and modified Whittle index policy are studied. The authors provide an LP based formulation of the relaxed problem and use a cutting plane algorithm to solve the Lagrangian dual problem.
Further, an information relaxation bound on the optimal value function is developed. 

Note that finding the optimal solution to a restless multi-armed
bandit problem is known to be PSPACE hard, \cite{Papadimitriou99}. However, the Whittle index policy  and LP based heuristic policies are  shown to be close to  optimal, \cite{Weber90,Bertsimas-Nino-Mora96, Bertsimas-Nino-Mora00}.  
Initial works on RMAB extensively studied models in which states of the arms are perfectly observable,\cite{Gittins11,mahajan08,verloop16,Nino-Mora01}.  Recently, 
 RMABs with partially observable states found interest
due to their applicability to problems in communication
networks,\cite{Nino-Mora08,LiuZhao10,LiNeely10}. 
For RMABs with partially observable states, the Whittle index policy
was shown to be nearly optimal in \cite{LiuZhao10,Ouyang16}.
 
 The myopic policy for solving RMAB has also been widely studied.  It has been shown to be optimal for some scenarios, \cite{Wang14,Zhao08,Murugesan12,Ahmad10}.

RMABs have been used for various applications in across domains.
Some specific applications include recommendation systems\cite{Meshram15,Meshram17}, 
sensor scheduling and target detection\cite{Nino-Mora11}, multi-UAV routing for observing targets\cite{Ny08}, stochastic network optimization\cite{LiNeely10}.
Most models assume instantaneous
feedback and their main interest is to study the Whittle-index or
myopic policy. 
An alternative index policy called as marginal
productivity index was studied by
\cite{Nino-Mora08,Nino-Mora09}. Marginal productivity index here, is
an extension of Whittle-index with interpretations from marginal
productivity theory used in economics.

The application of RMABs to decision making in cognitive radio networks was pioneered by Zhao \textit{et al.} in \cite{Zhao08}. Here, independent and identically evolving (i.i.d.) channels are considered; perfect state-observability is assumed when channels are used. It is shown that, myopic policy has a round-robin structure and is optimal for the case of two channels $(M=2).$ Meanwhile, \cite{Javidi08,Ahmad10} derived conditions for optimality of myopic policy for arbitrary number of i.i.d. channels $(M>2)$ under positively correlated conditions. Later, Liu and Zhao \cite{LiuZhao10} established the Whittle-indexability  for a class of RMABs which were applicable to problems of opportunistic access with perfect sensing. Further, they showed the optimality of  Whittle-index policy under certain conditions. In \cite{Wang14}, of the problem of multi-channel access with imperfect sensing and non-i.i.d. channels was studied; sufficient conditions for optimality of myopic policy were derived. Recently, \cite{Ouyang16} showed asymptotic optimality of Whittle-index policy for the downlink scheduling problem. 

The RMAB model has also been applied to the relay selection problem, \cite{Wei10}. In this work, perfect observability of channel states is assumed. The indices are computed by solving  a first order LP relaxation of RMAB developed in \cite{Bertsimas-Nino-Mora00}.

A common assumption in the above literature that deals with optimality of Whittle-index or myopic policies is the full observability of arm states when played. That is, in scenarios which allow more information to be gathered by playing arms, more general inferences can be made about the performance of various policies. This additional information also makes computation of Whittle-index expressions easier.
In recent work of \cite{Meshram18,Borkar17}, hidden Markov restless multi-armed bandit has been studied and Whittle-index policy is used.  This model assumes that arm state is never fully observable but only binary signals corresponding to each state transition are observed. In these models, partial observability of states makes it more challenging to prove indexability of arms and to obtain closed form index expressions.
 
Recall that in the current work we allow multiple state transitions in each decision interval; hence, information about arms states is even more sparser. This makes claiming indexability and deriving index expressions intractable in general. However, when the number of state transitions are quite large (although finite), tractability of the problem can be improved by making some assumptions (see Section~\ref{sec:sab}). The results provided are applicable for any finite  number of state transitions per decision interval.
\subsection{Contributions}
\begin{enumerate}
 \item We propose a novel methodology to solve the problem of sequential decision making with limited observation capability. We formulate this as restless multi-armed bandit with hidden states and cumulative feedback.
 \item The proposed 
 model allows multiple
   system state transitions during a decision interval. The feedback at the end of a decision interval is cumulative, that is, it represents all the state variations in the interval. This is a generalization of existing models which allow at most one transition.
   Hence, our model better represents rapidly varying channel conditions.   
   
 \item We analyze the problem by first studying single-armed LRBs.
  We show that for single armed LRBs, the optimal policy has a threshold structure.  Proving an optimal threshold policy is rendered cumbersome due to partially observability of states and cumulative feedback information. 
   In the process of showing an optimal threshold policy, we derive various structural properties of the action value functions for positively correlated arms  and negatively correlated arms. 
\item We prove the indexability of a class of lazy restless bandits with hidden states, under some restrictions on discount parameter. This is achieved by making use of the optimal threshold policy result and the properties of the action value functions. We also expect the result to hold for more general conditions; supported by a simulation study.
 \item We derive the closed form expression of Whittle-index for two special cases. Further, we present the Whittle-index computation algorithm for more general settings. This algorithm is based on a two-timescale stochastic approximation scheme. 
 
 \item In order to assess of the degree of sub-optimality of the Whittle-index policy, an upper bound based on the Lagrangian relaxation is provided. We present  an algorithm to compute this bound; it is based on stochastic finite difference method.
 
  We also provide a discussion on RMABs as a subset of the weakly coupled Markov Decsion processes (WC-MDP) and interpret some results from the literature in context of the current model.
  
\item  An extensive comparative study of the Whittle-index policy with other policies such as modified Whittle-index policy, myopic, uniform random, non-uniform random and round robin is provided.

\end{enumerate}

The rest of this document is organized as follows. The system model description is given in Section~\ref{sec:model}, where, an optimization problem for lazy restless bandits is formulated. Single-armed LRBs are analysed in Section~\ref{sec:sab}. Procedures for  Whittle-index computation are provided in Section~\ref{sec:Whittle-index-compute}.
 A discussion is WC-MDP is given in Section~\ref{sec:wcmdp} along with an  algorithm to compute the Lagrangian relaxation bound. A numerical study is provided in Section~\ref{sec:sim} before concluding in Section~\ref{sec:conclusion}.

\section{Model Description and Preliminaries}
\label{sec:model}

Let us consider a \textit{lazy restless multi-armed bandit} with $M$
independent arms. The time-line is divided into sessions that are
indexed by $s$. Each arm represents a channel/link in a communication
system. We model each channel using a Markov chain. Each arm has two states, say, good ($1$) and bad ($0$). At
any arbitrary time, each arm exists in one of the two states.
$Y_m(s) \in \{0,1\}$ denotes the state of arm $m$ at the beginning of
session $s.$ 
Let $K (\geqslant 1)$ be the number of state transitions for each arm
in a given session; it is finite and known to the decision maker. The state of each arm evolves according to a Markov chain. $p^m_{i,j}$ represents the transition probability of
arm $m$ from state $i$ to state $j,$ $i,j \in \{0,1\}$ and the
corresponding transition probability matrix (TPM) is denoted by $P_m=
[[p_{i,j}]].$ In a given session $s,$ the decision maker plays one arm
out of $M$ arms. $A_m(s)$ denotes the action corresponding to arm $m$
in session $s.$ If arm $m$ is played in session $s,$ then $A_m(s)=1$
and $A_m(s) = 0,$ otherwise. Since only one arm is played in a
session, $\sum_{m=1}^{M} A_m(s)=1.$

A reward is accrued at the end of each session from the arm that is played.
It depends on both initial and intermittent states of the arm, during the session. We denote $R_{m,i}$
as the average reward from playing arm $m$ which is in state $i$ at the beginning of the session. No rewards are accrued from arms that are not played. Rewards are not observable by the decision maker. Instead, a feedback is received from the arm that is played, at the end of the session in the form of ACK(1) or NACK(0). An
ACK means a successful session and a NACK means a failed
session. $Z_m(s) \in \{0,1\}$ denotes the feedback signal that is
obtained at end of session $s$ if arm $m$ is played in session $s.$
This feedback is probabilistic. We define $\rho_{m,i} :=
\mathsf{Pr} \{ Z_m(s)=1 \mid A_m(s)=1,Y_m(s)=i\},$ $i \in \{ 0,1\},$ which is the probability of success from playing arm $m$ which is in state $i$ at the beginning of the session.

The values of $R_{m,i}$ are independent of time(session). The average reward from  playing an arm is same in different sessions if the arm begins these sessions in the same state.  
Further, rewards $R_{m,0},$ $ R_{m,1}$ depend on the number of transitions $K;$ and are constant for a given value of $K.$ Similar is the case for probabilities of success $\rho_{m,0}, \rho_{m,1}.$ We do not use any additional notation to emphasize this dependence, as we assume $K$ is known and constant. 

An important assumption here is that, the ordering on rewards $R_{m,i}$  is same as the ordering on success probabilities $\rho_{m,i}.$ That is, if $R_{m,0}<R_{m,1},$ then $\rho_{m,0} < \rho_{m,1},$ and vice-versa. So, greater average reward means greater success probability.

Now, the exact state of each arm is not observable by the
decision maker. The decision maker maintains a belief about the
state of each arm. Let $\pi_m(s)$ the probability that arm $m$ is in
state $0$ at the beginning of session $s$ given the history $H(s),$
where $H(s) = \{A(l),Z(l)\}_{1\leq l<s}.$ Thus $\pi_m(s) :=
\prob{Y_m(s)=0 \mid H(s)}.$ The belief $\pi_m(s)$ about arm $m,$ is
updated by the decision maker at the end of every session $s$, based
on the action taken $A_m(s)$ and feedback received $Z_m(s)$.

Let $\phi := \{\phi(s) \}_{s \geq 0}$ be the policy, where $\phi(s):
H(s) \rightarrow \{1, \cdots, M\}$ maps the history up to session $s$
to action of playing one of the $M$ arms.  Let $A_m^{\phi}(s) = 1,$ if
$\phi(s) = m,$ and $A_m^{\phi}(s) = 0,$ if $\phi(s) \neq m.$
The expected reward from playing arm $m$ in session $s$ is $R_m(\pi_m(s)) \coloneqq \pi_m(s) R_{m,0} + (1-\pi_m(s))R_{m,1}.$
The infinite horizon expected discounted reward under policy $\phi$ is
given by
\begin{eqnarray}
  V_{\phi}(\pi) := \mathrm{E} \bigg\{ \sum_{s=1}^{\infty} \beta^{s-1}
  \sum_{m=1}^{M} A_m^{\phi}(s) \left( \pi_m(s) R_{m,0} \nonumber \right.\\
   \left. + (1-\pi_m(s)) R_{m,1} \right) \bigg\}.
  \label{eqn:val-fun-opt}
\end{eqnarray}
Here, $\beta$ is discount parameter, $0<\beta <1$ and the initial
belief $\pi = [\pi_1, \cdots, \pi_M],$ $\pi_m := \prob{Y_m(1)= 0 }.$
Our objective  is to find the policy $\phi$ that maximizes
$V_{\phi}(\pi)$ for all $\pi \in [0,1]^{M}.$
 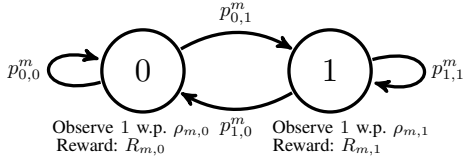
\begin{figure}
 
   \begin{center}
   \resizebox{0.35\textwidth}{!}{
     \begin{tikzpicture}[draw=black,>=stealth', auto, ultra thick, node distance=2cm]
       \tikzstyle{every state}=[fill=white,draw=black, ultra thick,text=black,scale=1.6]
       \node[state]    (A)                     {$0$};
       \node[state]    (B)[ right of=A]   {$1$};
       \path
       (A) edge[loop left]     node{$p_{0,0}^m$}
       (A) edge[bend left,above,->]      node{$p_{0,1}^m $ } (B)
       (B) edge[loop right]    node{$p_{1,1}^m$}
       (B) edge[bend left,below,->]      node{$p_{1,0}^m$} (A); 
       \draw (-0.4,-1) node {\small{ No Signal} };
       \draw (-0.3,-1.3) node {\small{No Reward }};
       \draw (3.3,-1) node {\small{ No Signal }};
       \draw (3.4,-1.3) node {\small{No Reward }}; 
     \end{tikzpicture}}
   \end{center}
   %
   
   
  \centerline{Arm $m$ is not played in session $s$ ($A_m(s) = 0$)}
   
   %
   \begin{center}
     \resizebox{0.35\textwidth}{!}{
     \begin{tikzpicture}[draw=black, >=stealth', auto, ultra thick, node distance=2cm]
     \tikzstyle{every state}=[fill=white,draw=black, ultra thick,text=black,scale=1.6]
       \node[state]    (A)                {$0$};
       \node[state]    (B)[ right of=A]   {$1$};
       \path 
       (A) edge[loop left]     node{$p_{0,0}^{m}$}         (A)
       edge[bend left,above,->]      node{$p_{0,1}^{m}$}      (B)
       (B) edge[loop right]    node{$p_{1,1}^{m}$}     (B)
       edge[bend left,below,->]      node{$p_{1,0}^{m}$}       (A);    
       \draw (-0.2,-1) node {\small {Observe $1$ w.p. $\rho_{m,0}$} };
       \draw (-0.6,-1.3) node {\small{ Reward: $R_{m,0}$} };
       \draw (3.5,-1) node {\small{ Observe $1$ w.p. $\rho_{m,1}$} };
       \draw (3.1,-1.3) node {\small{ Reward: $R_{m,1}$} };
     \end{tikzpicture}}
     
     
     \centerline{Arm $m$ is played ($A_m(s) = 1$)}
   \end{center}
   
   \caption{The state transition probabilities, the reward, and the
     probability of ACK $(1)$ being observed are illustrated above when
	the arm is not played. Also, the corresponding quantities are
	illustrated below when the arm is played. \label{fig:Arm-MC}}  

	\end{figure}
In \cite{Whittle88}, Lagrangian relaxation of this problem is analyzed  by 
introducing subsidy for not playing the arm. A solution to the relaxed problem is obtained by first studying the  
single-armed restless bandit.

\section{Single-armed Lazy Restless Bandit}
\label{sec:sab}
We consider a subsidy $\eta$ assigned if the arm is not played.  As
we are dealing with a single arm, we drop notation $m$ for
convenience.  In the view of subsidy $\eta$ one can reformulate
problem in~\eqref{eqn:val-fun-opt} for single-armed bandit as follows.
%
%
\begin{eqnarray}
  V_{\phi}(\pi): = \mathrm{E}\bigg\{\sum_{s=1}^{\infty} \beta^{s-1}
  \bigg( A^{\phi}(s) \left( \pi(s) R_0 +  (1-\pi(s)) R_1 \right) \nonumber \\+ \eta (1 - A^{\phi}(s)) \bigg) \bigg\} 
\end{eqnarray}
The goal is to find the policy $\phi$ that maximizes $V_{\phi}(\pi)$
for $\pi \in [0,1],$ $\pi$ is the initial belief.

We now describe the belief update rules. They determine the properties of the value functions. The following expressions can be obtained by employing Bayes rule.
\begin{enumerate}
\item If a channel is used for transmission in session $s$ and ACK
  is received, i.e., $A(s)=1$ and $Z(s)=1,$ then the belief at the
  beginning of session $s+1$ is $\pi(s+1) = \gamma_1(\pi(s)).$
  Here,
%
  \begin{eqnarray*}
    \gamma_1(\pi(s)):=\frac{(1-\pi(s))\rho_{1}p_{1,0} + \pi(s)
      \rho_{0}p_{0,0}}{\rho_{1}(1-\pi(s)) + \rho_{0}\pi(s)}.
  \end{eqnarray*}
%
\item If a channel is used for transmission in session $s$ and NACK
  is received, i.e., $A(s)=1$ and $Z(s)=0,$ then the belief at the
  beginning of session $s+1$ is $\pi (s+1) = \gamma_0(\pi(s)),$
  where
%
      \begin{eqnarray*}
        \gamma_0(\pi(s)) := \frac{(1-\pi(s))(1-\rho_{1})p_{1,0} +
          \pi(s)(1-\rho_{0}) p_{0,0}}{(1-\rho_{1})(1-\pi(s))
          + (1-\rho_{0})\pi(s)}.
      \end{eqnarray*}
%
\item If a channel is not used for transmission, i.e., $A(s)=0,$
  then the belief at the beginning of session $s+1$ is $\pi(s+1) =
  \gamma_2(\pi(s)),$ where
  {\small{
      \begin{dmath}
        {\gamma_2(\pi(s)) := \left(p_{0,0}-p_{1,0} \right)^{K} \pi(s) +
        p_{1,0} \frac{ \left(1 - (p_{0,0}-p_{1,0})^K \right)}{1-
          (p_{0,0}-p_{1,0})}.}
        \label{eqn:ksteps}
      \end{dmath}
  }}
This is because the channel is evolving independently,  after $K$ transitions of channel state, we  obtain belief as given in  the expression~\eqref{eqn:ksteps}.
\end{enumerate} 

Arms with $p_{0,0} > p_{1,0},$ are
called positively correlated as they tend to cling to their current state and evolve gradually.
Whereas, negatively correlated arms $p_{0,0} < p_{1,0},$ tend to change states more frequently.  
\begin{table}[h]
	\centering
	\caption{Numerical examples when difference between $p_{00}-p_{10} =
		0.2,0.5$ and $K=5,10$} \scalebox{0.85}{
		\begin{tabular}{|c|c|c|c|c|c|c|}
			\hline
			$p_{00}$ & $p_{10}$ & $\rho_0$ & $\rho_1$   & $K$  & $\gamma_2(\pi)$ & $q$ \\   \hline 
			$0.9$ & $0.4$& $0$ & $0.95$ & $10$ & $0.80$ & $0.8$ \\ 
			\hline
			$0.95$ & $0.45$ & $0$ & $0.95$ & $10$ & $0.9$  & $0.9$ \\ 
			\hline
			$0.8$  & $0.3$ & $0.2$ & $0.95$ & $10$ & $0.6$ & $0.6$  \\ 
			\hline 
			$0.8$ & $0.6$ & $0.2$ & $0.95$ & $5$ & $0.75$  & $0.75$ \\ \hline 
			$0.5$ & $0.3$ & $0.1$  & $0.9$ & $5$ &$0.375$ & $0.375$  \\ \hline  
		\end{tabular} }
		\label{table:evolution-gamma2}
	\end{table} 

\begin{remark}
	We can see from the expression of $\gamma_2$ in
	Eqn.~\eqref{eqn:ksteps} that for fixed value of $\pi,$ as $K
	\rightarrow \infty,$ we get $\gamma_2(\pi) \rightarrow q,$ where $q
	= \frac{p_{1,0}}{1- (p_{0,0}-p_{1,0})}.$ As mentioned earlier, in this work we assume $K$ to be finite, although it may be arbitrarily large. The rate of convergence of
	$\gamma_2$ to $q$ depends on $(p_{0,0}-p_{1,0}).$ This suggests that
	for large values of $K,$ we can approximate $\gamma_2(\pi)$ with $q.$
	If $\vert p_{0,0} - p_{1,0}\vert$ is smaller, then $K$ required for
	this approximation is small.
	In Table.~\ref{table:evolution-gamma2}, we present few examples
	where 1) $ \vert p_{0,0}-p_{1,0} \vert = 0.5,$ then $\gamma_2(\pi)
	\approx q$ for $k = 10,$ and 2) $ \vert p_{0,0}-p_{1,0} \vert =
	0.2,$ then $\gamma_2(\pi) \approx q$ for $k = 5.$
\end{remark} 
We now seek a stationary deterministic
policy. From~\cite{Ross71,Lovejoy87}, we
know that $\pi(s)$ is a sufficient statistic for constructing such
policies and the optimal value function can be determined by solving
following dynamic program.

{\small{
		\begin{eqnarray}
		V_S(\pi) &=& R_S(\pi) + \beta \left(\rho(\pi) V(\gamma_1(\pi)) + 
		(1-\rho(\pi)) V(\gamma_0(\pi)) \right) \nonumber \\
		V_{NS}(\pi)&=& \eta + \beta V(\gamma_2(\pi)) \nonumber \\
		V(\pi) &=& \max\{ V_S(\pi), V_{NS}(\pi) \}.
		\label{eqn:dynprog-sab}
		\end{eqnarray} 
	}}
	Here, $R_S(\pi) = \pi R_0 + (1-\pi) R_1,$ and $\rho(\pi) = \pi \rho_0 + (1-\pi)\rho_1. $
		
	Given belief $ \pi, $ $ V_S(\pi) $ denotes the value (discounted cumulative reward) of the decision to play the arm in current session and then follow the optimal policy for all future sessions. Similarly, $ V_{NS}(\pi) $ is the value of the decision to not play the arm in current session and then follow the optimal policy for all future sessions.
	
	We next derive structural results for these value functions. 
	A sketch of proof is provided along with each result. Detailed proofs for major results can be found in the Appendix.
	The following lemma is about the convexity of value functions in belief $\pi$ and  subsidy $\eta.$

\begin{lemma}
\label{lemma:valfunc-convex}
\begin{enumerate}
\item For fixed $\eta,$ $V_S(\pi),$ $V_{NS}(\pi)$ and $V(\pi)$ are
  convex in $\pi.$
\item For fixed $\pi,$ $V_{S}(\pi,\eta),$ $V_{NS}(\pi,\eta)$ and
  $V(\pi,\eta)$ are non-decreasing and convex in $\eta.$
\end{enumerate}
\end{lemma}
\textit{Sketch of proof:}
1)  We know that applying value iteration on an initial set of functions would produce sequences of functions $V_{S,n},$ $V_{NS,n},$ $V_n.$ Let $V_{S,1}(\pi)=R(\pi),$ $V_{NS,1}(\pi)=\eta$ and $V_1=\max\{V_{S,1},V_{NS,1}\};$ all of which are convex.  Now, we assume $V_{S,n},$ $V_{NS,n}$ are convex and show that $V_{S,n+1},$ $V_{NS,n+1},$ are convex. Then, by induction $V_{S,n},$ $V_{NS,n}$ and $ V_n $ are convex for all $ n.$ We know that these sequences of functions converge uniformly to $V_S,$ $V_{NS}$ and $V$ by value iteration. And the result follows.\\
2) Proof of the second part can also be claimed using similar argument.
The detailed proof is given in Section A,B of the Appendix.

We first provide structural results for positively correlated arms. Later we will study negatively correlated arms.
\subsection{Positively correlated arm}
	Let us start by looking at the properties of belief update functions; they will be useful while proving the computing the Whittle index.
\begin{lemma}
For positively correlated arm, i.e., $p_{0,0} > p_{1,0},$
  the belief updates $\gamma_0(\pi),$ $\gamma_1(\pi)$ and
  $\gamma_2(\pi)$ are increasing in $\pi.$ Further, $\gamma_1(\pi)$
  and $\gamma_0(\pi)$ are convex and concave, respectively. Also,
  $p_{1,0}\leq \gamma_1(\pi) \leq \gamma_0(\pi) \leq p_{0,0}.$
\label{lemma:gammas-prop} 
\end{lemma}
 The proof can be claimed by looking at the signs of the first and second derivatives of these functions. 
\begin{lemma}
 For a positively correlated arm $(p_{0,0}>p_{1,0})$ with a fixed subsidy $\eta,$ $\beta \in (0,1)$,the
value functions $V(\pi),$ $V_S(\pi)$ and $V_{NS}(\pi)$ are decreasing
in $\pi.$
\label{lemma:valfunc-dec-for-positivecorr}
\end{lemma}
	\textit{Sketch of proof:} Assume that $V_{S,n},V_{NS,n}$ and $V_n$ are non increasing in $\pi.$ We need to show, $V_{n+1}(\pi)\geq V_{n+1}(\pi')$ for $\pi'>\pi.$ We know $R_S(\pi)>R_S(\pi').$ Let $u = [V_n(\gamma_0(\pi)), V_n(\gamma_1(\pi))],$  $ v = [1-\rho(\pi),\rho(\pi)]^T $ and $v'=[1-\rho(\pi'),\rho(\pi')]^T.$ We see  stochastic ordering $v'\leqslant_s v.$ Hence, $uv \geq uv'$ by the property of stochastic ordering $(\leqslant_s).$ 
	After some algebra it follows that $V_{n+1}(\pi) \geq V_{n+1}(\pi')$.
	Similarly, we can argue for $V_{NS}(\pi).$

See Section C of Appendix for detailed proof.

	Notice that the difference between action value functions,  $V_S(\pi)-V_{NS}(\pi)$ gives the advantage of playing over not playing, for belief state $\pi.$ We will show that this function is decreasing in belief.
\begin{lemma}
For fixed subsidy $\eta,$ and $p_{0,0} > p_{1,0}.$
the function $(V_S -V_{NS})(\pi)$ is
decreasing in $\pi$ for any of the following conditions
\begin{enumerate}
\item For large $K,$ \textit{i.e.} $\gamma_2(\pi) \approx q,$
\item For any $K>1,$ when, $0<p_{0,0} - p_{1,0}<\frac{b}{5}$ and $\beta \in (0,1),$ where, $b = \min\left\lbrace 1, \frac{R_1-R_0}{\rho_1-\rho_0}\right\rbrace,$
\item For any $K>1,$ when, $\beta \in (0,b/5).$ 
\end{enumerate}
\label{lemma:dpi-decreasing-positivecorr}
\end{lemma}
\textit{Sketch of proof:} Part 1) For large $K,$ the approximation $\gamma_2(\pi)\approx q$ makes $V_{NS}(\pi)$ independent of $\pi.$ Further, from Lemma~\ref{lemma:valfunc-dec-for-positivecorr}, $V_S$ is decreasing in $\pi$ for positively correlated arms. Hence, the result follows.\\
	 For Parts 2) and Part 3), the key ideas involved in the proof are as follows. (1) We first bound the derivatives of $V_S,$ $V_{NS}$ and $V$ w.r.t. $\pi,$ (see Lemma 9 in Section D of the Appendix). This is also called as the Lipschitz property of the value functions. The Lipschitz constant is explicitly calculated. (2) Then, we show that the derivative of $V_S-V_{NS}$ w.r.t. $\pi$ is negative under the given conditions. 
	  One might consider the right partial derivatives at points where any of the functions are non-differentiable.
For an arbitrary $K,$ we claim a decreasing advantage of playing by imposing conditions on either the transition probabilities or the discount factor.

The detailed proof can be found in Section D of the Appendix.

The following lemma gives the properties of belief update functions for negatively correlated arms. These properties are complementary to those of positively correlated arms. That is, if $\gamma_{(.)}(\pi)$ is increasing, convex for positively correlated arms, it is decreasing, concave for negatively correlated arms.

\begin{lemma}
For negatively correlated arm, i.e., $p_{0,0}<p_{1,0} $
  the belief updates $\gamma_0(\pi),$ $\gamma_1(\pi)$  are decreasing in $\pi.$ 
  Further, $\gamma_1(\pi)$  and $\gamma_0(\pi)$ are concave and convex, respectively. Also,
  $p_{0,0}\leq \gamma_0(\pi) \leq \gamma_1(\pi) \leq p_{1,0}.$
\label{lemma:gammas-prop-ve} 
\end{lemma}
The proof can be claimed by looking at the signs of the first and second derivatives of these functions.

The advantage of playing, $(V_S -V_{NS})(\pi)$ decreases with belief, even in case of negatively correlated arms. 
\begin{lemma}
\label{lemma:dpi-decreases-negativecorr}
For fixed subsidy $\eta,$ and $p_{0,0}<p_{1,0},$
the difference function $(V_S -V_{NS})(\pi)$ is
decreasing in $\pi$ under any of the following conditions
\begin{enumerate}
\item For any $K>1,$ when, $0<p_{1,0} - p_{0,0}<\frac{b}{5}$ and $\beta \in (0,1),$
\item For any $K>1,$ when, $\beta \in (0,b/5),$ where, $b = \min\left\lbrace 1, \frac{R_1-R_0}{\rho_1-\rho_0}\right\rbrace.$
\end{enumerate} 
\end{lemma}
\begin{remark}
\begin{itemize}
\item  Note that for negatively correlated arms, $V_S$ is not necessarily decreasing in $\pi,$ unlike their positively correlated counterparts. Hence, it is difficult to prove that $V_S-V_{NS}$ is decreasing in $\pi,$ even for part 1 of Lemma~\ref{lemma:dpi-decreases-negativecorr} (the case of large $K,$ $\gamma_2(\pi)\approx q.$)  
\item So, the same lengthy procedure used for proving Lemma~\ref{lemma:dpi-decreasing-positivecorr}-part $2,3$ is needed to prove all the parts of Lemma~\ref{lemma:dpi-decreases-negativecorr}.
\end{itemize}
\end{remark}
\subsection{Threshold policy and Indexability}
We now define a threshold type policy and we will show
that the optimal policy is threshold type  for single armed bandit.
\begin{definition}
A policy is called as a threshold type for single armed bandit if
there exists $\pi_T \in [0,1]$ such that an optimal action is to play
the arm if $\pi \leq \pi_T$ and to not play the arm if $\pi \geq
\pi_T.$
\end{definition}
\begin{theorem}
For fixed subsidy $\eta,$ $\beta \in (0,1),$ the optimal policy for
single-armed bandit is of a threshold type for each of the following
conditions.
\begin{enumerate}
 \item If $K$ is large i.e., $\gamma_2(\pi) \approx q.$
 \item For any $K\geq 1,$ if $0<p_{0,0} - p_{1,0}<b/5.$
 \item For any $K\geq 1,$ if $0<p_{1,0} - p_{0,0}<b/5.$ 
 \item For any $\beta \in (0,b/5),$ where, $b = \min\left\lbrace 1, \frac{R_1-R_0}{\rho_1-\rho_0}\right\rbrace.$
\end{enumerate}
\label{thm:threshold-policy-1}
\end{theorem}
\begin{IEEEproof}
From the preceding Lemma~\ref{lemma:dpi-decreasing-positivecorr} and Lemma~\ref{lemma:dpi-decreases-negativecorr}, we know that
$(V_{s}(\pi) - V_{NS}(\pi))$ is a decreasing in $\pi.$ Further,
$V_S(\pi)$ and $V_{NS}(\pi)$ are convex in $\pi.$ This implies that
there exists a either $\pi_T \in [0,1]$ such that $V_S(\pi_T) =
V_{NS}(\pi_T)$ or $V_S(\pi) > V_{NS}(\pi)$ for all $\pi,$ or
$V_{S}(\pi) < V_{NS}(\pi)$ for all $\pi.$ This leads to desired
result.
\end{IEEEproof}

We expect that the optimal policy is of threshold type even when the conditions in Theorem~\ref{thm:threshold-policy-1} are not valid. However, it is difficult to prove this in general. We illustrate this threshold structure for general conditions $\rho_0<\rho_1,R_0<R_1$ and $\beta\in(0,1),$ using a numerical example in Section~\ref{sec:sim-threshold}.

We here define the indexability and will show that a
single-armed bandit is indexable.  Using exact threshold-type policy
result, we define the following.
\begin{eqnarray*}
  \mathcal{P}_{\beta}(\eta) := \left\{\pi \in [0,1]: 
  V_{S}(\pi,\eta) \leq V_{NS}(\pi,\eta)  \right\}.
\end{eqnarray*}
It is a set of belief state $\pi$ for which the optimal action is to not
to play the arm,  i.e., $A(s) = 0.$ From~\cite{Whittle88}, we  state the
definition of indexability.

\begin{definition}
A single-armed restless bandit is indexable if
$\mathcal{P}_{\beta}(\eta)$ is monotonically increases from
$\emptyset$ to entire state space $[0,1]$ as $\eta$ increases from
$-\infty$ to $\infty,$ i.e., $\mathcal{P}_{\beta}(\eta_1) \setminus
\mathcal{P}_{\beta}(\eta_2) = \emptyset$ whenever $\eta_1 \leq
\eta_2.$
\label{defn:indexable}
\end{definition}
\begin{figure}[h]
\centering
\begin{center}
\begin{tikzpicture}[scale = 1.0]
\draw[|-|,red, thick] (0,0)  to (3.8,0);
\draw[-|,blue, thick] (3.8,0)  to (7,0);
\draw[->,black] (7,0)  to (8,0);
\node (a) at (0,-0.3) {0};   
\node (d) at (7.0,-0.3) {$1$};
\node (e) at (8,-0.3) {$\pi$};
\node (b) at (3.8,-0.3) {$\pi_T$};
\node (j) at (2.2,0.3) {`Play'};
\node (h) at (5.4,0.3) {`Not play'};
\end{tikzpicture}
\vspace{-0.4cm} 
\caption{In a threshold type policy, the optimal action shifts from `play' to `not play' at $\pi_T.$ If $\pi_T$ moves left as the subsidy $\eta$ increases, then the bandit can be called indexable. }
 \end{center}
\label{fig:threshold}
\end{figure}
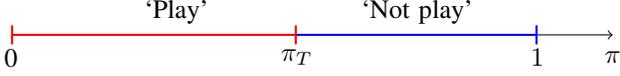
To show indexability, we require to prove that a threshold $\pi_T$ is a
monotonic function of $\eta$. We  state 
the following lemma from~\cite{Meshram18}.
\begin{lemma}
  Let $\pi_T(\eta) = \inf \{ \pi \in [0,1]: V_S(\pi, \eta) =
  V_{NS}(\pi, \eta) \}.$ If
 $   \frac{\partial V_{S}(\pi, \eta)}{\partial \eta}\bigg\vert_{\pi =
      \pi_T(\eta)} < \frac{\partial V_{NS}(\pi,\eta)}{\partial
      \eta}\bigg\vert_{\pi = \pi_T(\eta)}, $
  then $\pi_T(\eta)$ is monotonically decreasing function of $\eta.$
\label{lemma:pi-T-monotone-eta}
\end{lemma}
\textit{Sketch of proof:}{Assume that  $ \pi_T(\eta) < \pi_T(\eta^{\prime})$ for $\eta < \eta^{\prime}.$ For given $\eta$ and $\pi_T(\eta),$ we have $V_S(\pi_T(\eta)) = V_{NS}(\pi_T(\eta)).$   Using this we obtain $V_S(\pi_T(\eta^{\prime}))-V_{NS}(\pi_T(\eta^{\prime})) \geq 0 $ at  $\eta^{\prime} = \eta +\epsilon$ for some $\epsilon \in (0,c),$ $c<1.$ This implies $   \frac{\partial V_{S}(\pi, \eta)}{\partial \eta}\bigg\vert_{\pi =
		 	\pi_T(\eta)} > \frac{\partial V_{NS}(\pi,\eta)}{\partial
		 	\eta}\bigg\vert_{\pi = \pi_T(\eta)}. $ This contradicts our if statement, thus we have $\pi_T(\eta)$ is monotonically decreasing function of $\eta.$}

Note that the value function may not be differentiable as function of
$\eta;$ the right partial derivative is used in this case. It
exists due to convexity of value function in $\eta$ and rewards are
bounded.

We now use Definition~\ref{defn:indexable} and 
Lemma~\ref{lemma:pi-T-monotone-eta} to show that a single-armed
restless bandit in our setting is indexable under the conditions in Theorem~\ref{thm:threshold-policy-1}.
\begin{theorem}
If $\rho_0 < \rho_1,$ $R_0 < R_1,$ and $\beta \in (0,1/3),$ then, 
a single-armed restless bandit
 is indexable .
\label{thm:indexability-general}
\end{theorem}
\textit{Sketch of proof:} First, we bound the derivatives of $ V_S,$ $ V_{NS} $ and $ V $ with respect to $ \eta.$ This bound is given as $ \frac{1}{1-\beta} .$ Then, we show that for $\beta \in (0,1/3),$ the conditions for monotonicity of $ \pi_T(\eta) $ in $ \eta ,$ given in Lemma~\ref{lemma:pi-T-monotone-eta} are satisfied. 

The detailed proof can be found in the Section E of the Appendix.
We believe that the indexability result is true more generally, where,
we do not require any assumption on $\beta.$ This restriction on
$\beta$ is required here because of difficulty in obtaining
closed-form value function expression. But, for specific conditions
such as $\rho_0 = 0, \rho_1 =1,$ and $K >1,$ we can derive
closed-form expressions of value functions and we can obtain
conditions for indexability without any assumption on $\beta.$ 
We illustrate indexability of arms under more general conditions using a numerical example in Section~\ref{sec:sim-threshold}.
\section{Whittle index calculations for special cases}
\label{sec:Whittle-index-compute}
We first define the Whittle-index and later we provide index formula. In the following, we use $V_S(\pi,\eta),V_{NS}(\pi,\eta)$ instead of $V_S(\pi),V_{NS}(\pi),$ to emphasize their dependence on subsidy $\eta.$
{{
\begin{definition}[\cite{Whittle88}]
  If an indexable arm is in state $\pi,$ its Whittle-index 
  $W(\pi)$ is
  \begin{eqnarray}
    W(\pi)  &=& \inf \{\eta \in \mathbb{R}: V_{S}(\pi,\eta) =
    V_{NS}(\pi,\eta) \}.
    \label{eqn:Whittle-index}
  \end{eqnarray}
\label{def:whittleind}
\end{definition} }} 
\vspace*{-0.5cm}
The basic idea used for this computation is as follows. If the optimal policy has threshold structure, then, the Whittle-index $W(\pi)$ for belief $\pi$ is the subsidy required such that $\pi$ is the threshold, \textit{i.e.} $V_S(\pi,W(\pi))=V_{NS}(\pi,W(\pi)).$ So, we are required to obtain the action value function expressions, equate them and solve for $W(\pi).$ These expressions are obtained using the following idea.
We know that for any $\pi'>\pi,$ $V(\pi') = V_{NS}(\pi')$ and for every $\pi'<\pi,$ $V(\pi')=V_S(\pi').$ Now, we use the recursive definition of  action value functions given by the dynamic program in \eqref{eqn:dynprog-sab}. Then, we use the properties of the belief update functions $\gamma$'s to evaluate the expressions.
 
We provide expressions of Whittle-index for positively correlated
arms, \textit{i.e.}, $p_{0,0} > p_{1,0},$. We do this for two
special cases.
\begin{enumerate}
\item Arbitrary $K,$ $\rho_0=0$ and $\rho_1=1.$ 
\item $K$ is large, \textit{i.e.}, $\gamma_2(\pi) \approx q,$ $R_0 = \rho_0 = 0,$ and $0<R_1 = \rho_1 <1.$ 
\end{enumerate}
Recall that, $K$ is known to the decision maker and it does not vary.
Also, if $\rho_0=0$ and $\rho_1=1$ for an arm, an ACK would mean that the session started in state $1$ and a NACK would mean that it started in state $0.$  On the other hand, for the case $\rho_0 = 0$ and $\rho_1 <1 $, we can conclude from an ACK that the arm was in state $1$ at the session beginning; no such a conclusion cannot be made from a NACK.
For general cases we provide an algorithm to compute the index. It is
motivated from stochastic approximation algorithms.
We consider four intervals, $A_1,$ $A_2,$ $A_3,$ and $A_4,$ as shown in
Fig.~\ref{fig:areas-for-WI}; we compute the index for each interval
separately. These intervals were derived on the basis of the direction in which a belief $\pi$ is pulled by the different belief update functions $\gamma_0,$ $\gamma_1$ and $\gamma_2.$ The derivations for the following expressions can be found in Section F of the Appendix. 
\begin{figure}[h]
\centering
\begin{center}
\begin{tikzpicture}[scale = 1.0]
\draw[->, thick] (0,0)  to (8,0);
\draw[red, thick] (0.3,-0.1)  to (0.3,0.1);
\draw[red, thick] (7.0,-0.1)  to (7.0,0.1);
\draw[blue, thick] (1.8,-0.1)  to (1.8,0.1);
\draw[blue, thick] (3.2,-0.1)  to (3.2,0.1);
\draw[blue, thick] (5.8,-0.1)  to (5.8,0.1);
\node (a) at (0.3,-0.3) {0};   
\node (d) at (7.0,-0.3) {$1$};
\node (e) at (8,-0.3) {$\pi$};
\node (g) at (1.8,-0.3) {$p_{1,0}$};
\node (h) at (3.2,-0.3) {$q$};
\node (j) at (5.8,-0.3) {$p_{0,0}$};
\node (k) at (1,0.6) {$A_1$};
\node (l) at (2.5,0.6) {$A_2$};
\node (m) at (4.6,0.6) {$A_3$};
\node (o) at (6.3,0.6) {$A_4$};
\draw[<-, red, thick] (0.3,0.3)  to (0.7,0.3);
\draw[->, red, thick] (0.7,0.3)  to (1.8,0.3);
\draw[<-, blue, thick] (1.8,0.3)  to (2.5,0.3);
\draw[->, blue, thick] (2.5,0.3)  to (3.2,0.3);
\draw[<-, magenta, thick] (3.2,0.3)  to (3.8,0.3);
\draw[->, magenta, thick] (3.8,0.3)  to (5.8,0.3);
\draw[<-, red, thick] (5.8,0.3)  to (6.2,0.3);
\draw[->, red, thick] (6.2,0.3)  to (7.0,0.3);
\end{tikzpicture}
\caption{The different cases to calculate $W(\pi).$ }
 \end{center}
\label{fig:areas-for-WI}
\end{figure}
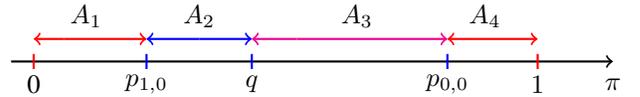
 \vspace*{-0.75cm} 
\subsection{Whittle-index for Case $1)$: arbitrary $K,$ $p_{0,0} > p_{1,0},$ $\rho_0=0, \rho_1=1.$} 
\begin{enumerate}
 \item For $\pi \in A_1,$  
 {\small{  
 \begin{equation*}
  W(\pi) = R_S(\pi) = R_1 + \pi(R_0-R_1).
 \end{equation*} }}
 \item For $\pi \in A_2,$
 {\small{
 \begin{equation*}
  W(\pi) = \frac{R_S(\pi)(1-\beta)\left[1-\beta(\pi-p_{1,0}) \right]}{1 - \beta[1+(1-\beta)(\pi-p_{1,0})]}.
 \end{equation*} }}
 \item For $\pi \in A_3$
 {\small{
 \begin{equation*}
  W(\pi) = \frac{D(\pi) -\beta D(\gamma_2(\pi))}{1+ \beta B(\gamma_2(\pi) ) -B(\pi)},
 \end{equation*} }} where, 
 {\small{
 $B(\pi) = \beta c [\pi(1-b) +b],$ and $ D(\pi) = R_S(\pi)+ \beta[(1-\pi)(a+bd) + \pi d]$
 \[a = \frac{R_S(p_{1,0})}{1-\beta(1-p_{1,0})}, b = \frac{\beta p_{1,0}}{1-\beta(1-p_{1,0})},\] \[a_1 = \frac{\beta^{t}R_S(\gamma_2^t(p_{0,0}))}{1-\beta^{t+1}\gamma_2^t(p_{0,0})} , b_1 = \frac{\beta^{t+1}(1-\gamma_2^t(p_{0,0}))}{1 - \beta^{t+1}\gamma_2^t(p_{0,0})} \]
 \[c = \frac{f}{1-bb_1}, d = \frac{a_1 +b_1a}{1-bb_1}, f = \frac{1-\beta^t}{(1-\beta)(1-\beta^{t+1}\gamma_2^t(p_{0,0}))}\] }}
 \item For $\pi \in A_4$,
 {\small{
 \begin{eqnarray*}
  W(\pi) = m\pi + c_1 -\beta(m\gamma_2(\pi) + c_1); \\ m = \frac{R_0-R_1}{1-\beta(p_{0,0}-p_{1,0})}, c_1 = \frac{R_1 + m\beta p_{1,0}}{1-\beta}. 
 \end{eqnarray*} }}
\end{enumerate}

\subsection{Whittle-index for Case $2)$:  Large $K,$ $\gamma_2(\pi) \approx q.$}
 
Here, we assume that $p_{0,0} > p_{1,0},$ and $K$ is large, that is $\gamma_2(\pi) \approx q,$ 
$R_0 = \rho_0 = 0,$ and $0<R_1 = \rho_1 <1.$ 

The index formula for each interval is given as follow. 
\begin{enumerate}
\item For $\pi \in A_1,$ 
the Whittle-index  $ W(\pi) = \rho(\pi).$
\item  For $\pi \in A_2,$  we consider following cases.
{
\begin{enumerate}
\item if $\gamma_0(p_{1,0}) \geq \pi,$ then, the Whittle-index is  {\small{
\begin{eqnarray*}
W(\pi) = \frac{\rho(\pi) }{1-\beta (\rho(p_{1,0})- \rho(\pi))}.
\end{eqnarray*} }}
\item if $\gamma_0(p_{1,0}) < \pi$
  but $\gamma_0^2(p_{1,0}) \geq \pi$ then,  Whittle-index 
$W(\pi) = \frac{\rho(\pi)}{C_1}.$
Here,
{\small{ 
\begin{eqnarray*}
C_1 = 1- \beta (\rho(p_{1,0}) - \rho(\pi)) - \beta^2
(\rho(\gamma_0(p_{1,0})) - \rho(\pi)) \\+  \beta^2
\rho(\gamma_0(p_{1,0})) \rho(p_{1,0}).
\end{eqnarray*}
}}
\end{enumerate}
}
\item For $\pi \in A_3,$ obtaining index is tedious, and this has to
  be computed numerically by using Algorithm~\ref{algo:WI}.
\item  For $\pi \in A_4,$  the  Whittle-index is, 
{\small{
\begin{eqnarray*}
W(\pi) = m \pi (1 - \beta (p_{0,0}- p_{1,0})) + (1-\beta)c - \beta
p_{1,0} m,
\end{eqnarray*}
}} where $
m = \frac{-\rho_1}{ 1 - \beta (p_{0,0}- p_{1,0})},$ and $c = \frac{\rho_1 + \frac{-\beta p_{1,0} \rho_1}{ 1 - \beta (p_{0,0}- p_{1,0})}}{1-\beta }.$
%

\end{enumerate}

We omit the derivation of expressions for Case $2$ due to space constraints.
 They can derived using similar procedure as Case $1.$ 
%
%
\subsection{Algorithm for Whittle-index computation}
We now present an algorithm for computing Whittle-index in a general
case. Here, for a given $\pi \in [0,1],$ assume that it is the
threshold and compute index $W(\pi).$ Start at $t=0$ with an initial
subsidy $\eta_0$ and run the value iteration algorithm to compute
action value functions $V_S(\pi,\eta_0)$ and $V_{NS}(\pi,\eta_0).$ The
subsidy $\eta_0$ incremented or decremented proportionally with the
difference $V_S(\pi,\eta_0)-V_{NS}(\pi,\eta_0)$ and a learning
parameter $\alpha,$ as follows
\[\eta_{t+1} = \eta_t + \alpha(V_S(\pi,\eta_t)-V_{NS}(\pi,\eta_t)).\]
The algorithm terminates when the difference
$|V_S(\pi,\eta_t)-V_{NS}(\pi,\eta_t)|<h,$ where $h$ is the tolerance
limit. See Algorithm~\ref{algo:WI} below for details. Here, we use two
timescales, one for updating the subsidy and the other for updating
value functions. The $\alpha$ parameter is chosen such that, subsidy
$\eta_t$ is updated at a slower timescale compared to the value
iteration algorithm that computes $V_S(\pi,\eta_t)$ and
$V_{NS}(\pi,\eta_t).$
This is a two-timescale stochastic approximation algorithm and is
based on similar schemes studied in \cite{Borkar08,Borkar17b}. In
\cite[Chapter 6]{Borkar08}, the convergence of a two-timescale
stochastic approximation algorithm was discussed.  
{\footnotesize{
\begin{algorithm}
\KwIn{Reward values $R_0,R_1$;  Initial subsidy $\eta_0,$ tolerance $h,$ step size $\alpha.$ }
\KwOut{Whittle's index $W(\pi)$}
	for $\pi \in [0,1]$;\\
		$\eta_t \gets \eta_0$;\\
	\While {$|V_S(\pi)-V_{NS}(\pi)|>h$ } {
 	 $\eta_{t+1} = \eta_t + \alpha(V_S(\pi,\eta_t)-V_{NS}(\pi,\eta_t)); $\\
 	 $t=t+1;$\\
 	 compute $V_S(\pi,\eta_t)$, $V_{NS}(\pi,\eta_t);$
	 }
 \Return{$W(\pi) \gets \eta_t$};
\caption{{\sc } Whittle-index computation for single arm}
\label{algo:WI}
\end{algorithm}
}} \vspace{-0.5cm}

\vspace{5mm}
\section{Relaxations of Restless Multi-armed Bandits}
\label{sec:wcmdp}
A body of literature on Markov decision processes focuses on a linear programming (LP) approach to solving problems involving weakly coupled Markov decision processes (WC-MDP). In this section, we discuss results from the literature about the bounds on the optimal value function for WC-MDP and their implications for RMABs. The first bound comes from the Lagrangian relaxation of the WC-MDP problem, while the other comes from approximate dynamic programming (ADP). 
\vspace{-4mm}
\subsection{LP approach to POMDP}
A POMDP can be seen as an MDP with a uncountably infinite state space. The following is the dynamic program (DP) for a POMDP. A function $V$ that satisfies the following equation $\forall \bm{\pi}\in \mathcal{S}_\Pi$ is the optimal value function of the POMDP. Here, $\bm{\pi}$ is a belief state which is a vector with components that sum to $1,$ and $\mathcal{A}_{\bm{\pi}}$ is its set of actions. Let $\mathcal{S}_\Pi$ be the belief space which is a simplex and $\mathcal{S}_{O}$ is the set of possible observations.
\begin{eqnarray}
 V(\bm{\pi}) = \max\limits_{\bm{a}\in \mathcal{A}_{\bm{\pi}}}\left\lbrace R_{\bm{a}}(\bm{\pi}) + \beta\sum\limits_{\bm{o}\in\mathcal{S}_{O}} V(\Gamma_o(\bm{\pi}))\Pr(\bm{o}|\boldsymbol{\pi,a})\right\rbrace
\end{eqnarray}
An LP formulation of the above DP is as follows. Given $\nu(\bm{\pi})>0,\forall\bm{\pi}\in \mathcal{S}_{\Pi},$ 
\begin{eqnarray}
 &\min\limits_{V(.)} &\int\limits_{\mathcal{S}_{\Pi}}\nu(\bm{\pi})V(\bm{\pi})d\bm{\pi} \nonumber\\
 s.t\texttt{ }
 &V(\bm{\pi})&\geq R_a(\bm{\pi}) + \beta\sum\limits_{o\in\mathcal{S}_{O}} \Pr(\bm{o}|\bm{\pi},\bm{a})V(\Gamma_{\bm{o}}(\bm{\pi})),\nonumber\\ 
 &\forall & a\in\mathcal{A}_{\bm{\pi}},\bm{\pi}\in\mathcal{S}_{\Pi}  .
\end{eqnarray}
Note that this program is infinite dimensional LP, \cite{Lerma96}[Chapter $6$].

\subsection{Weakly coupled POMDPs (WC-POMDPs)}
In this set up there are a set of POMDPs which are coupled together through a set of linear constraints. Clearly, RMABs with partially observable states fit in this scenario; each arm is a POMDP.
	 
Let $\bm{\pi}$ be a belief state, $\bm{o}$ is observation, $\mathcal{S}_{O}$ is the set of observations. In case of WC-POMDPs, a belief state $\bm{\pi}$ is an element of a polymatroid belief space $\mathcal{S}_\Pi.$ This is unlike the case of a POMDP where the belief space is a simplex. The feasible action set for belief state $\bm{\pi}\in \mathcal{S}_{\Pi}$ is 
\[\mathcal{A}_{\bm{\pi}} = \bigg\lbrace \bm{a}\in \{0,1\}^M :  \sum\limits_{m=1}^{M}{a_m} =N\bigg\rbrace.\] 
Note that $\mathcal{A}_{\bm{\pi}}$ here considers the coupling constraint specific to RMABs. In general WC-MDPs allow linear inequality constraints.
\subsection{Bounds for Weakly coupled POMDPs (WC-POMDPs)}
We now discuss bounds on the optimal value function for weakly coupled POMDPs and their implications for RMABs with hidden states. The first bound comes from the Lagrangian relaxation of the weakly coupled MDPs problem, while the other comes from approximate dynamic programming (ADP). These bounds are derived using linear programming (LP) formulations. 	
	The dynamic program for an RMAB with partially observable states that is  formulated as a weakly coupled POMDP problem is given by
	\begin{align}
	V(\bm{\pi})= \max\limits_{\bm{a}\in \mathcal{A}_{\bm{\pi}}} \bigg\lbrace \sum\limits_{m=1}^{M} & R_{a_m}(\pi_m)  +  \nonumber\\ &\beta\sum\limits_{\bm{o}\in\mathcal{S}_{O}} \Pr(\bm{o}|\bm{\pi},\bm{a})V(\Gamma_{\bm{o}}(\bm{\pi}))\bigg\rbrace 
	\label{eqn:weak-couple-dynamic}
	\end{align}
	for $\bm{\pi}\in \mathcal{S}_{\Pi}.$
	This can  be cast as an LP along with the coupling constraint.
	\begin{align}
	&H(\nu) = \min\limits_{V(.)}\int\limits_{\mathcal{S}_{\Pi}}\nu(\bm{\pi})V(\bm{\pi})d\bm{\pi} \nonumber \\
	& \text{s.t.} \texttt{ } V(\bm{\pi})\geqslant \sum\limits_{m=1}^{M}  R_{a_m}(\pi_m)  +  \beta\sum\limits_{\bm{o}\in\mathcal{S}_{O}} \Pr(\bm{o}|\bm{\pi},\bm{a})V(\Gamma_{\bm{o}}(\bm{\pi})) \nonumber\\
	& \forall \bm{\pi}\in \mathcal{S}_{\Pi}. 
	\end{align}
	\begin{remark}
		The infinite dimension linear program formulations for POMDP and RMAB are prohibitively complex to solve computationally. Hence, we can employ grid approximation of the belief space to bring tractability to the problem. The problem is then reduced to an POMDP with finite belief states; results from existing literature can be directly applied. Note that the set of observations remains unaltered by grid approximation.
	\end{remark} 
	In the following, we present some results from \cite{Hawkins03,Adelman08}, on Lagrangian and ADP relaxations bounds for the weakly coupled MDPs. We utilize them to assess the degree of sub-optimality of heuristic index policies used for  RMABs.

	We first describe the Lagrangian bound on the optimal value function of RMAB with hidden states, defined over a grid approximation of the belief state space. To avoid use of additional symbols, we utilize the same notation for the discretized  version of the problem; $\mathcal{S}_\Pi$ for the belief state space, $\bm{\pi}$ for belief, etc.

\subsubsection{The Lagrangian Bound} 
	The Lagrangian relaxation of dynamic program \eqref{eqn:weak-couple-dynamic} decouples the weakly-coupled POMDPs into optimization over single POMDPs, this is given as by 
	\begin{align}
	\label{finite_lagrange_rmab}
	V^{\lambda}(\bm{\pi}) = \max\limits_{\bm{a}} &\bigg\lbrace \sum\limits_{m=1}^{M}  R_{a_m}(\pi_m)  + {\lambda}\bigg(N- \sum\limits_{m=1}^{M}a_m \bigg)  \nonumber\\ &+ \beta\sum\limits_{\bm{o}\in\mathcal{S}_{O}}V^{\lambda}(\Gamma_{\bm{o}}(\bm{\pi}))\prod\limits_{m=1}^{M}\Pr(o_m|\pi_m,a_m)\bigg\rbrace\\
	& s.t.\texttt{ }a_m\in\{0,1\},m\in\{1,...,M\}\nonumber. 
	\end{align}
	The following lemma is \cite[Proposition 1]{Adelman08}; it states that the value function of the Lagrangian relaxed RMAB can be written as the summation of value functions for individual arms.
	\begin{lemma}
		\label{lemma:lagrange_decouple}
		\begin{align}
		\label{eqn:lagrange_decouple}
		V^{\lambda}(\bm{\pi}) = &\frac{N\lambda}{1-\beta} + \sum\limits_{m=1}^{M}V_m^{\lambda}(\pi_m),\\
		V_m^{\lambda}({\pi_m}) = &\max\limits_{{a}\in \{0,1\}} \bigg\lbrace R_{a_m}(\pi_m)  - {\lambda}a_m  \nonumber\\ &+ \beta\sum\limits_{{o_m}\in\mathcal{S}_{O_m}}V_m^{\lambda}(\Gamma_{{o_m}}({\pi_m}))\Pr(o_m|\pi_m,a_m)\bigg\rbrace. 
		\end{align}
	\end{lemma}
	Also, $V^{\lambda}(\bm{\pi})$ is convex and piecewise linear in $ \lambda. $ 
	Now, optimizing over Lagrangian variable $\lambda$ gives 
	\[ 	V^{\lambda^{\ast}}(\bm{\pi}) = \min_{\lambda \geq 0} V^{\lambda}(\bm{\pi}).	\]
	Further, the optimal value function $V(\bm{\pi}) \leq V^{\lambda^{\ast}}(\bm{\pi}).$ This provides the Lagrangian relaxation bound.
	One can use linear programming schemes to compute the bound, where optimization is taken over both $V(\cdot)$ and $\lambda;$ the LP formulation is as follows.
	\begin{align}
	\label{lagrange_bound_lp_primal}
	H^{\lambda^*}(\nu) &= \min\limits_{V(.),\lambda} \frac{N\lambda}{1-\beta} + \sum\limits_{m=1}^{M}\sum\limits_{\pi_m\in \mathcal{S}_{\Pi_m}}\nu(\pi_m)V_m(\pi_m),\nonumber\\ s.t.\texttt{ }& V_m(\pi_m)\geqslant R_{a_m}(\pi_m) -\lambda a_m \nonumber \\ &+ \beta\sum\limits_{{o_m}\in\mathcal{S}_{O_m}}V_m^{\lambda}(\Gamma_{{o_m}}({\pi_m}))\Pr(o_m|\pi_m,a_m)
	\end{align}
	 Alternatively, the Lagrangian bound can be computed using a stochastic finite difference scheme. We present this scheme in Section~\ref{Lb-comp} as Algorithm~\ref{algo:Lagrange_bound}.
	\subsubsection{The ADP bound} 
	We now discuss the bound that can be obtained from approximate dynamic programming. The LP approach to ADP was developed for MDPs in \cite{Farias03}. It was later extended to WC-MDPs in \cite{Hawkins03},\cite{Adelman08}.	
	In this formulation, WC-POMDPs are decoupled by employing the ADP method, where a linear approximation of the value function is used. The value function is of the form $V(\bm{\pi}) \approx \theta + \sum\limits_{m=1}^{M}V_m(\pi_m).$ The closest LP based approximation of the value function will be the solution of the following LP (see \cite[Section 2.4]{Adelman08}.
	\begin{align}
	\label{lp_adp_primal}
	H^{ADP}(\nu) &= \theta + \min\limits_{V(.)}\sum\limits_{m=1}^{M}\sum\limits_{\pi_m\in \mathcal{S}_{\Pi_m}}V_m(\pi_m), \nonumber\\
	s.t.&\texttt{ }\theta(1-\beta) + \sum\limits_{m=1}^{M}V_m(\pi_m) \geqslant \sum\limits_{m=1}^M R_{a_m}(\pi_m) +\nonumber \\& \beta\sum\limits_{m=1}^M \sum\limits_{{o_m}\in\mathcal{S}_{O_m}}V_m^{\lambda}(\Gamma_{{o_m}}({\pi_m}))\Pr(o_m|\pi_m,a_m),\nonumber\\
	&\forall \bm{\pi}\in\mathcal{S}_{\Pi}, \{\bm{a}\in \{0,1\}^M :\sum_{m=1}^M a_m = N\} .
	\end{align}
	One of the main findings in \cite{Adelman08} is $H(\nu)\leqslant H^{ADP}(\nu)\leqslant H^{\lambda^*}(\nu)$ for any $\nu \geqslant 0.$ This suggests that the LP based ADP bound is tighter than the Lagrangian bound. Moreover, bounds on the relaxation gaps $i.e.,$ the distance of each bound from the optimal is given as follows.
	
	{\small{ 
	\begin{equation*} 
	H^{\lambda(\cdot)}(\nu)-H(\nu)\leqslant \frac{(N+1)\mathcal{E} ^* + \Omega}{1-\beta}; 
	H^{ADP}(\nu)-H(\nu)\leqslant \frac{ \Omega '}{1-\beta}
	\end{equation*} 
	}}
	
	Here $\mathcal{E}^*,$ $\Omega,$ $\Omega '$ depend on problem parameters and $\mathcal{E}^*$ can be bounded by a constant. When $\Omega,$ $\Omega '$ are sub-linear in $M,$ the average relaxation gap per arm goes to zero as the number of arms $(M)$ increases.
	\begin{remark}
		
		In the preceding discussion, both the Lagrangian and ADP bounds uses linear programming algorithms for computation. 
		In our formulation, we are dealing with a grid approximation of an uncountable state space. Hence, as the number of arms increases, the number of variables and constraints becomes too large in LP which increases computational complexity.
		Stochastic sub-gradient scheme is an alternative which is computationally less expensive, but it may be slower to converge.
		Also note that for RMABs which play one arm at a time, both the Lagrangian bounds and ADP bounds are equal. 
	\end{remark}
\subsection{Computation of the Lagrangian Bound $({L_b})$}
\label{Lb-comp}
	We now present an algorithm for computation of the Lagrangian bound; it is based on stochastic finite difference scheme and value iteration. For a given multiplier $\lambda,$ the Lagrange relaxed value function is given in \eqref{eqn:lagrange_decouple}. Here, each component $V_m^{\lambda}$ is equivalent to the value function of the single armed bandit problem corresponding to arm $m;$ this can be computed using value iteration. As we are dealing with a continuous state space due to partial observability, uniform grid approximation of the belief space is considered. A variant of value iteration known as Gauss-Seidel value iteration (GSVI) is used for value function computation. GSVI converges faster than classical value iteration \cite{puterman14}, as it substitutes updated values for states as soon as they are computed. For the case of a POMDP this becomes Gauss-Seidel value approximation (GSVA), described in Algorithm~\ref{algo:GSVA}. Here, $\mathcal{S}_{\Pi_G}$ is the grid  approximated belief set with granularity $\delta,$ \textit{i.e.,} distance between successive belief points is $\delta.$ In grid approximation, the continuous belief space is mapped to the finite set $\mathcal{S}_{\Pi_G}$ using nearest neighbour approximation (NNA). Also recall $\mathcal{S}_{O}$ is the set of all observations including `no observation'.
	Now, for a given $\lambda,$ the Lagrange relaxed value function can be computed by employing Algorithm~\ref{algo:GSVA}, $M$ times - once for each arm.  To find the Lagrangian bound, we need to find $\lambda^*$ which minimizes this value in Eqn.~\eqref{eqn:lagrange_decouple}. This is achieved using a stochastic finite difference scheme which is described in Algorithm~\ref{algo:Lagrange_bound}. There are two steps involved in this scheme. First, we compute the value of the bound in  Eqn.~\eqref{eqn:lagrange_decouple} for a given $\lambda,$ by value iteration (GSVA). In the second step, we compute the finite difference approximation of the sub-gradient which is used to  update $\lambda.$ The stopping criterion is that estimated sub-gradient falls below the tolerance $ \delta. $  
	
	\textit{Convergence :} 
The quantity $g^{\lambda}_t$ computed in Algorithm~\ref{algo:Lagrange_bound} is a finite difference approximation (FDA) of the subgradient. It is a well known result (see \cite[Chapter 8]{Nocedal06}) that the subgradient can be written as the sum of its FDA  and an additive error $\delta_\epsilon.$ And, $\delta_\epsilon \rightarrow 0$ as $\epsilon \rightarrow 0,$ where $\epsilon$ is the denominator term of the FDA. In Algorithm 3, this can be ensured by assigning a small value to the parameter $\alpha.$	
	The convergence of Algorithm~\ref{algo:Lagrange_bound} to the global minima of the Lagrangian function in Eqn.~\eqref{eqn:lagrange_decouple} can be argued using the idea of two-timescale stochastic approximation. Notice that in computation of the bound, there are two time scales (discrete) involved. Along one time scale the value $ V^{\lambda_t} $ of the bound is updated by value iteration while keeping $ \lambda_t $ constant. Along the second time scale, the update of $ \lambda_t $ happens. Hence, the second time scale is slower compared to the first. It is a well known result in stochastic approximation that, such two-timescale algorithms converge if the sequence $ \alpha_t $ is decreasing, $ \sum_t \alpha_t = \infty $ and $\sum_t \alpha_t^2 <\infty.$ This convergence is almost sure as shown in \cite[Chapter $6,$ Theorem $2$]{Borkar08}. If $ \alpha_t $ is replaced with  a small constant value $ \alpha, $ there is convergence with high probability. For details, see \cite[Chapter $9,$ Section $9.3$]{Borkar08}.
	
	\begin{algorithm}[t]
		\KwIn{$R_{m,0}^a,R_{m,1}^a,a\in\{0,1\}$,$\eta$, $\rho_{m,0},\rho_{m,1}$, $P_m,$  for $m=1,...,M$; belief update functions $\Gamma 's$;
			tolerance $h,$ discount factor $\beta,$ step size $\alpha,$ grid granularity $\delta.$ }
		\KwOut{ $V(\pi),\forall \pi\in\mathcal{S}_{\Pi_G}$}
		\textit{initialization}\ $t=0,$ $V^0=\frac{\rho_{m,0}}{1-\beta}$ ;\\
		\While{}
		{
			\For{$\pi\in \mathcal{S}_{\Pi_G}$} 
			{Find $S^{\pi}_{\leqslant} = \{o\in \mathcal{S}_{O},a\in \mathcal{A}_{\pi}|\Gamma^a_o(\pi)\leqslant \pi\},$\\
				$S^{\pi}_{>} = \{o\in \mathcal{S}_{O},a\in \mathcal{A}_{\pi}|\Gamma^a_o(\pi)> \pi\},$\\
				Compute \\ $\hat{\Gamma^a_o}(\pi)=NNA(\Gamma^a_o(\pi),\mathcal{S}_{\Pi_G}),$ $\forall o\in \mathcal{S}_O,a\in \mathcal{A}_{\pi},$\\
				$r^0(\pi)=\eta,$ $r^1(\pi)=R(\pi)$\\
				Compute {\footnotesize{
						\begin{align*}		
						V^{t+1}(\pi) \leftarrow \max\limits_{a\in\mathcal{A}_{\pi}}\bigg\{ r^a(\pi)&+\beta\bigg[ \sum\limits_{o\in S^{\pi}_{\leqslant}}\Pr(o|\pi,a)V^{t+1}(\hat{\Gamma^a_o}(\pi)) \\&+ \sum\limits_{o\in S^{\pi}_{>}}\Pr(o|\pi,a)V^{t}(\hat{\Gamma^a_o}(\pi)) \bigg]\bigg\},
						\end{align*} }}
				$\pi\leftarrow \pi + \delta,$\\
			}
			\eIf{$\lVert V^{t+1}-V^t \rVert_1 \leqslant h$}
			{$V\leftarrow V^{t+1},$\\ 		
				break;}
			{ 	 	
				$V^t\leftarrow V^{t+1}$\\
				$t\leftarrow t+1,$\\
				continue;}
		}
		\Return $V$
		\caption{Gauss-Seidel Value Approximation (GSVA)  for POMDP }
		\label{algo:GSVA}
	\end{algorithm} 
	\begin{algorithm}
		\KwIn{$R_{m}$, $\rho_{m}$, $P_m$ for $m=1,...,M$; belief update functions $\gamma 's$;
			Initial Lagrange multiplier $\lambda_0,$ tolerance $\delta,$ discount factor $\beta,$ step sizes $\alpha_t.$ }
		\KwOut{Lagrangian bound $V^{\lambda^*}(\pi),\lambda^*$}
		\textit{initialization}\ $t=1,$ $\lambda_t = \lambda_0,$ $V^{\lambda}=\frac{N }{1-\beta}\min\{R_{m,0}\}$ ;\\
		\While{}
		{
			Compute $V_m^{\lambda_t}\leftarrow GSVA(R_m,\lambda_t,\rho_m,P_m,\Gamma,\beta),$ for $m=1,..,M,$\\
			$V^{\lambda_t}\leftarrow \frac{N{\lambda_t}}{1-\beta} + \sum\limits_{m=1}^{M}V_m^{{\lambda_t}},$ \\
			$g^{\lambda}_t \leftarrow \frac{V^{\lambda_t}-V^{\lambda}}{\lambda_t-\lambda_{t-1}},$\\
			\eIf{$\lvert g^{\lambda}_t \rvert \leqslant \delta $}
			{$V^{\lambda^*}\leftarrow V^{\lambda_t},\lambda^*\leftarrow \lambda_t ,$\\ 		 	 break;}
			{ 	 	
				$V^{\lambda}\leftarrow V^{\lambda_t},$\\
				$\lambda_{t+1}\leftarrow \lambda_t + \alpha_t g^{\lambda}_t,$\\
				$t\leftarrow t+1,$\\
				continue;}
		}
		\Return{$V^{\lambda^*},\lambda^* $} \\
		\caption{Lagrangian bound $(L_b)$ computation for RMAB}
		\label{algo:Lagrange_bound}
	\end{algorithm}
\section{Numerical Simulations and Discussion}
\label{sec:sim}
\subsection{Threshold policy and Indexability}
\label{sec:sim-threshold}
We begin by illustrating the threshold type structure of optimal policy and indexability of LRBs through an example. The parameters used are $p_{0,0}=0.2,$ $p_{1,0}=0.9,$ $R_0 = \rho_0 = 0.3,$ $R_1 = \rho_1 = 0.9,$ $K=3$ and $\beta=0.99.$ In Fig.~\ref{fig:threshold_and_indexability} we plot the action value functions  $V_S(\pi)$ and $V_{NS}(\pi)$ for two values of subsidy $\eta.$ Although the parameters used do not satisfy the conditions in Theorem~\ref{thm:threshold-policy-1}, the optimal policy has threshold structure. As the subsidy increases, the threshold moves towards the left, thus increasing the set of beliefs for which not-playing is the optimal action; this means the arm is indexable, by definition. We conjecture that lazy restless bandits are indexable under conditions $\rho_0<\rho_1,$ $R_0<R_1$ and $\beta\in(0,1).$
\begin{figure}[h]
  \begin{center}
    \begin{tabular}{cc cc}
     \hspace{-0.4 cm}
      \includegraphics[scale=0.155]{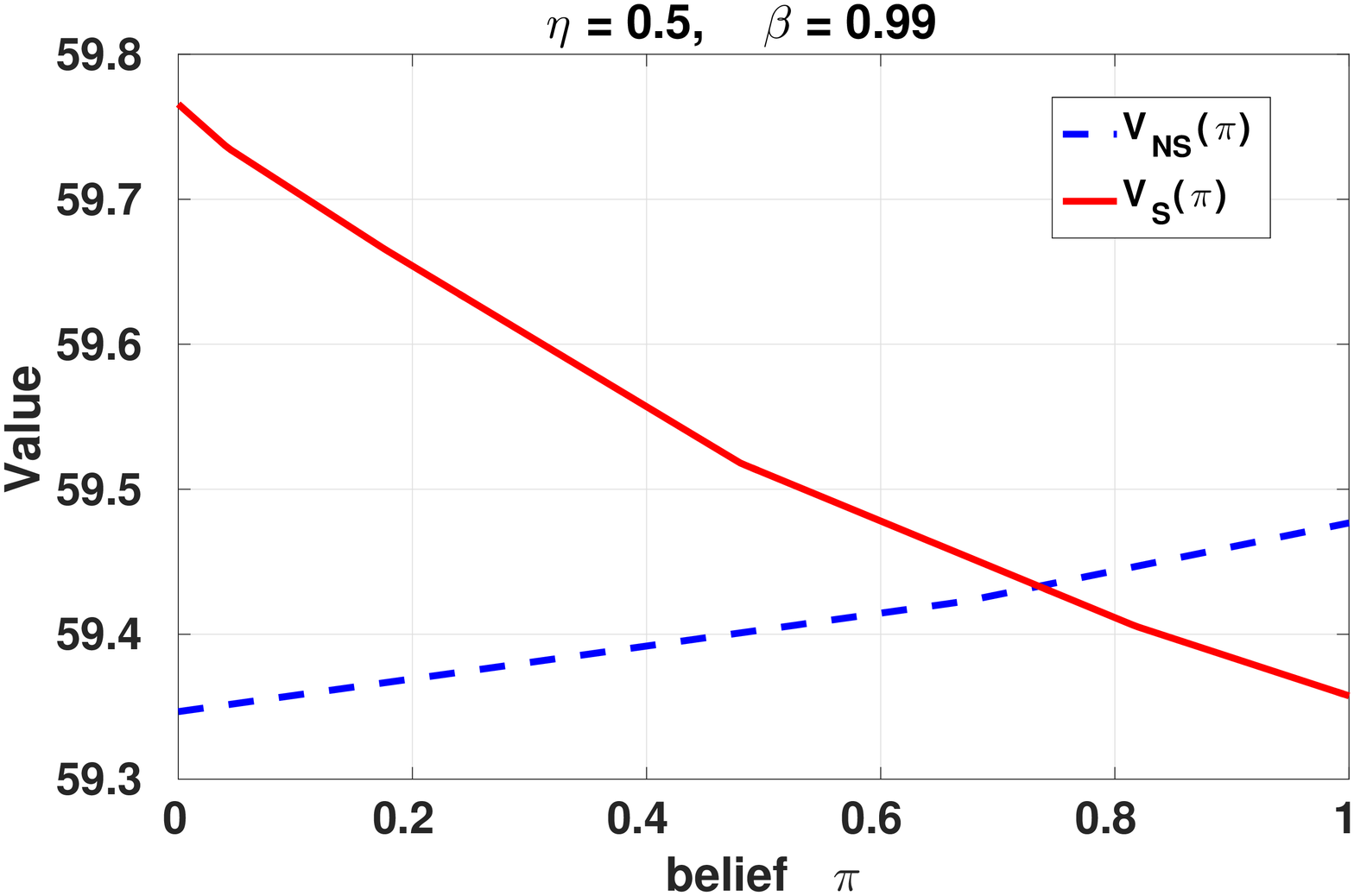}
      & 
      \hspace{-0.6 cm}
      \includegraphics[scale=0.155]{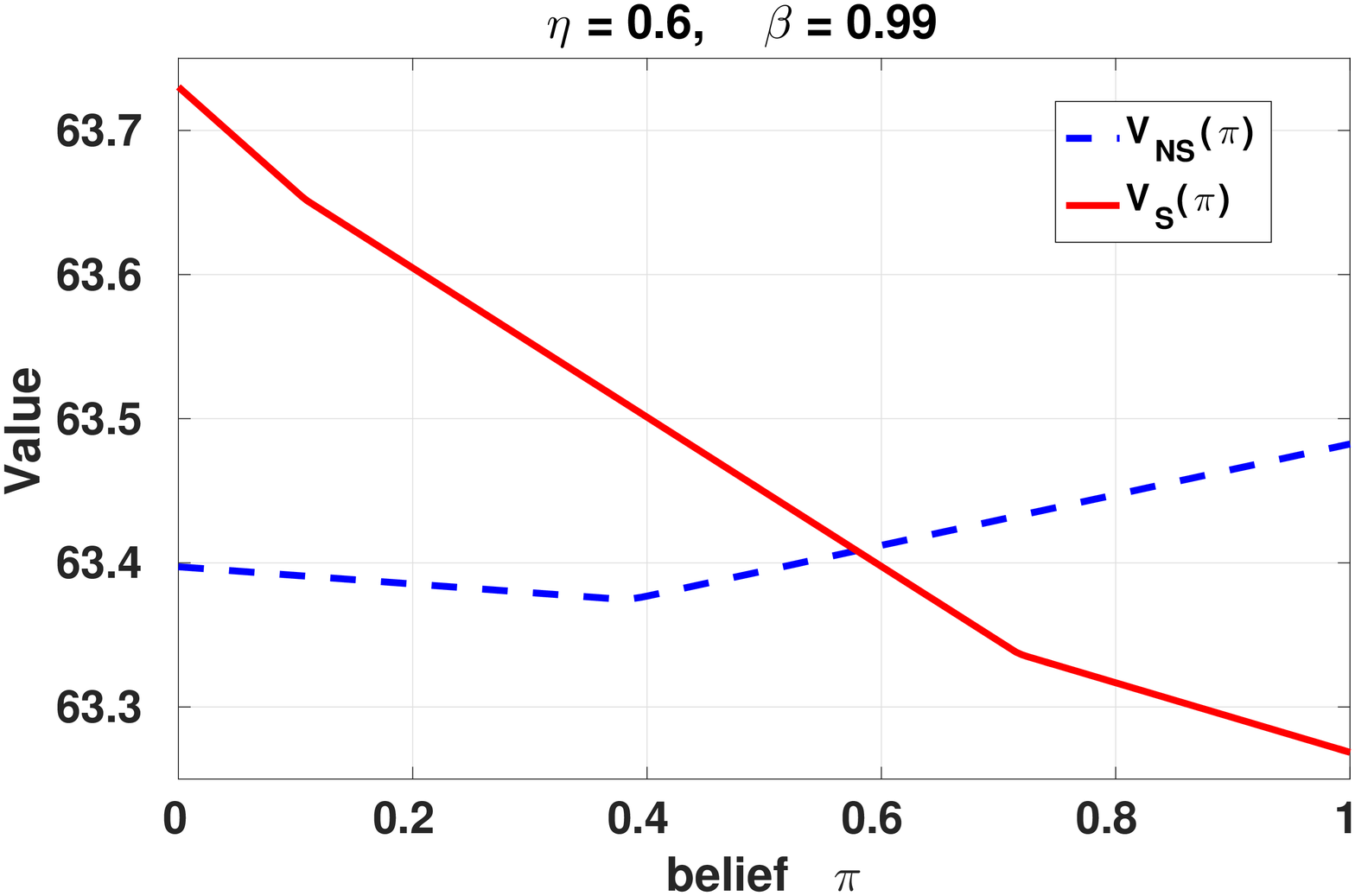} \\
      \hspace{-0.2cm}a) {\small{ For $\eta=0.5,$ $\pi_T(\eta)=0.72$}}   & \hspace{-0.4cm} b) {\small{ For $\eta=0.6,$ $\pi_T(\eta)=0.58$}}  
         \end{tabular}
  \end{center}
  \caption{ The optimal choice of action switches from playing to not-playing at a belief threshold $\pi_T(\eta),$ where, $V_S(\pi_T(\eta))=V_{NS}(\pi_T(\eta)).$ The threshold moves left as subsidy $\eta$ is increased from $0.5$ in $a)$ to $0.6$ in $b)$, implying indexability of the arm.}
  \label{fig:threshold_and_indexability}
\end{figure}
%
%
We now present a few numerical examples and compare different
policies.  The
policies compared are 1) Whittle-index
policy (WI)-- plays the arm with highest Whittle-index in each session, 2) modified Whittle index policy (MWI) - plays the arm with highest modified Whittle index 3) myopic policy (MP)-- plays arm with highest immediate expected reward in each session, 4)
uniformly random (UR), 5) non-uniform random (NUR)-- plays arm
randomly with distribution derived from current belief and 6) round
robin (RR)-- plays arm in round robin order. Further, we will also compare the value obtained by these policies to the Lagrangian upper bound on the optimal value.

\textbf{\textit{Modified Whittle index}}: Modified Whittle index (MWI) defined in \cite[Section $4.3$]{Brown17} is an alternative to Whittle index that is less computationally complex. It was defined for a finite horizon MDP, and can be computed recursively through a series of Bellman operations. Hence, these indices depend on both state and time, unlike the Whittle index which depends only on the state. For a single armed bandit, MWI $m_t(\pi)$ at time $t$ for belief state $\pi$ is the difference between action value functions (for playing and not playing) computed till that time. We need to compute MWI for large time horizons $(T=500,1000)$ to provide for a fair comparison with Whittle index which is defined for infinite time horizon.

MATLAB was used for performing simulations. In these simulations, the
arms start in a random state with a given initial belief about the
state of the arm.  In each session one arm is played according to the
given policy of study. Reward from the played arm is accumulated stored at the end of each session. Later, these rewards are averaged over $L$
iterations (sample paths of states).

We shall compare the discounted cumulative rewards 
obtained from each of the policies as a function of session number. Another parameter of interest while comparing various policies is the
arm choice fraction which is defined as follows.  Let $1_{m,s,l}$ be the indicator
variable if arm $m$ is played in session $s,$ and $l^{th}$
iteration. Then $N_{m,l} := \frac{1}{S_{\max}} \sum_{s=1}^{S_{\max} }
1_{m,s,l},$ where $S_{\max}$ number of sessions for which simulations
are performed. Further, this fraction is averaged over $L$ iterations.
We call this as the choice fraction of arm $m$ corresponding to the policy under study. 

We illustrate five numerical examples. 
In first two examples, we assume that $K$ is
large, \textit{i.e.,} $\gamma_2(\pi)= q$.
For the
last three examples we have a more general setting. We compare $\%$
value gain of various policies with uniform random policy as the
baseline.
First, we shall look at an example which compares the current model that considers multiple transitions $K>1$ per decision interval to the $K=1$ model in \cite{Meshram18} that allows only transition. In order to make this comparison, whenever we use $P$ as the transition matrix for an arm with the current model, $P^K$ will be the corresponding matrix for that arm with the $K=1$ model.
\vspace{-0.3 cm}

\subsection{Example-0 : A six armed bandit}
We choose $\rho_0=0,$ $\rho_1=1$ for all arms, $K=10,$
\begin{table}[h]
\begin{center}
\begin{tabular}{ccccccc}
\hline
$p_{0,0}=$ &$0.7$ & $0.6$ & $0.5$ & $0.8$ & $0.6$ & $0.8$ \\
$p_{1,0}=$ &$0.2$ & $0.2$ & $0.2$ & $0.3$ & $0.3$ & $0.6$ \\
$R_0 =$ & $0.1$ & $0.25$ & $0.3$ & $0$ & $0.15$ & $0.2$ \\
$R_1 =$ & $1$ & $0.85$ & $0.8$ & $1$ & $0.95$ & $0.9$ \\
\hline
\end{tabular}
\end{center}
\end{table}
\vspace{-0.5cm}

Whittle indices were computed using the expressions presented in Section~\ref{sec:Whittle-index-compute}. Fig.~\ref{fig:reward-armchoice-Ex0} shows that the current model gives a better cumulative reward compared to previous model in \cite{Meshram18}. It can also be seen that the arm choices turn out to be different. 
\begin{figure}[h]
  \begin{center}
    \begin{tabular}{cc cc}
    \hspace{-0.5cm}
      \includegraphics[scale=0.155]{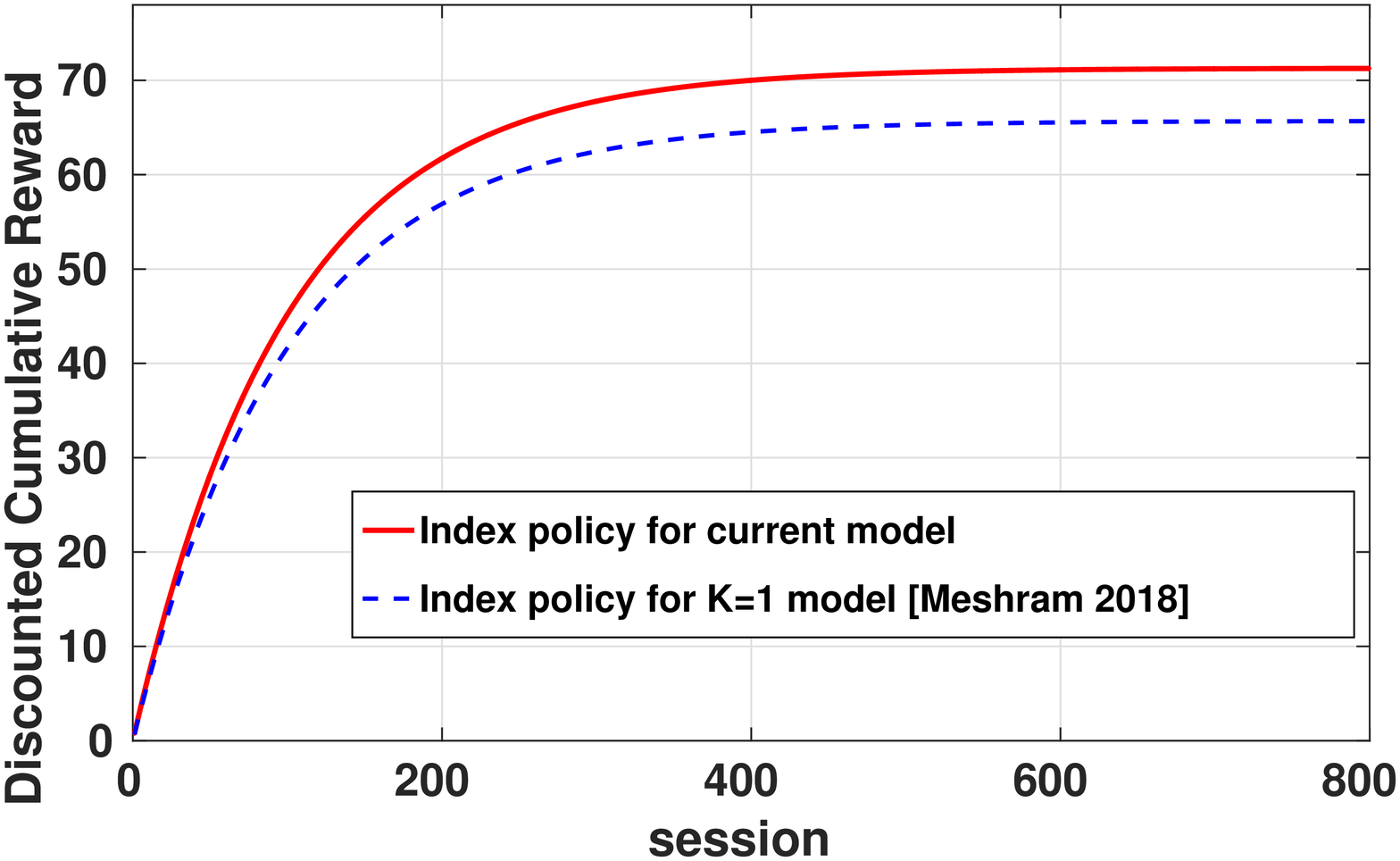}
      & 
      \hspace{-0.7cm}
      \includegraphics[scale=0.155]{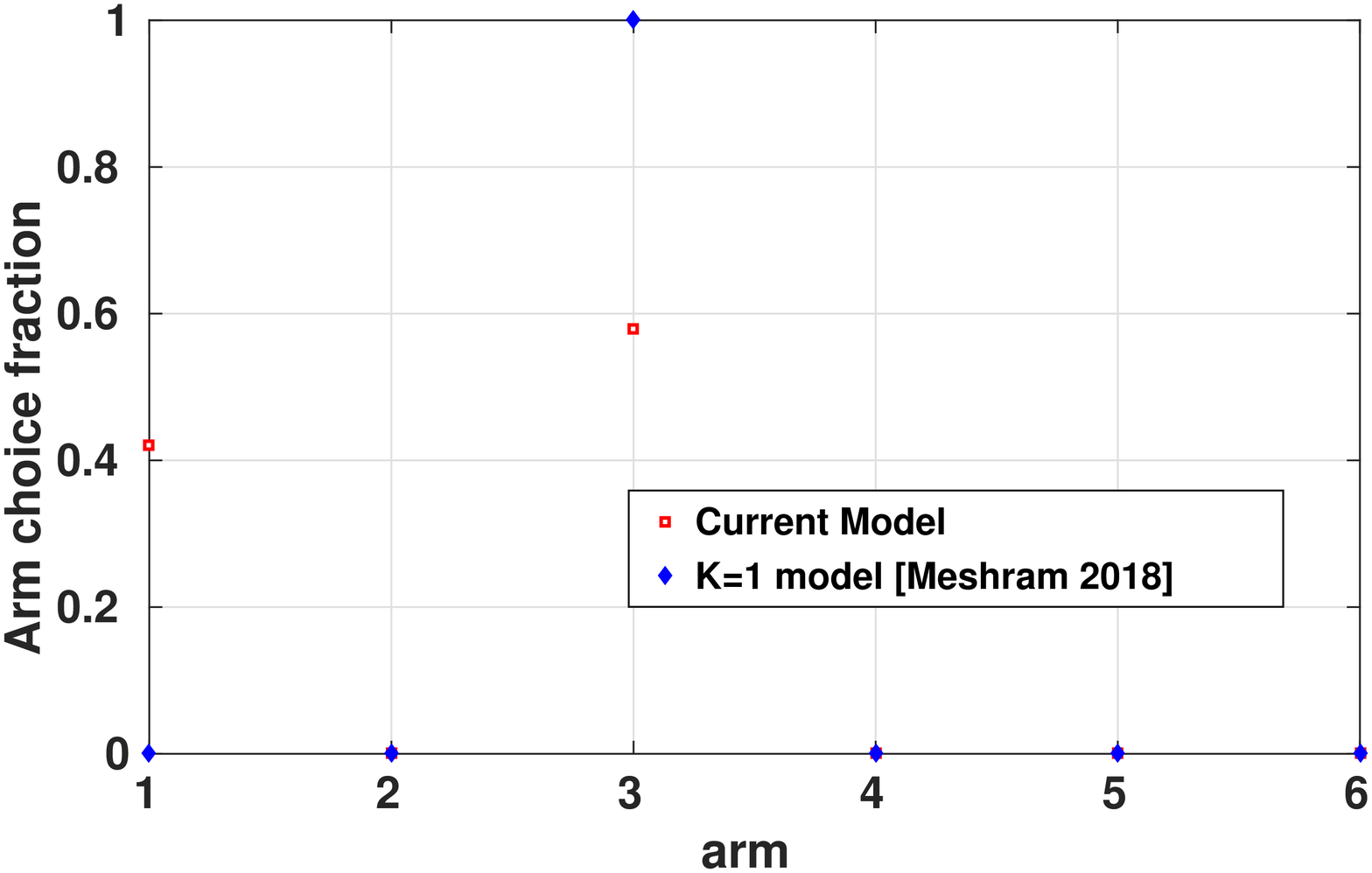} \\
      a) {\small{ Discounted cumulative reward}}   & \hspace{-0.7cm} b) {\small{ Arm choice fraction }}  
         \end{tabular}
  \end{center}
  \vspace{-0.3 cm}
  \caption{ Example-0: a) discounted cumulative rewards as function of
    sessions for Whittle policy from current model and previous model\cite{Meshram18}, b) arm choice fraction for
    each arm. Number of transitions per session $K=10.$}
  \label{fig:reward-armchoice-Ex0}
\end{figure} 
%
%
\subsection{Example-1 : Arms with similar reward structure and stationary behavior }
In this scenario, all the arms have identical reward from play of that
arm and $K$ is large. Also, all the arms have same $q_m = 0.45,$
except for arm $9,$ i.e. $q_9 = 0.4.$ We use following set of
parameters: $\rho_0=R_0=0,$ $\rho_1 =R_1= 0.9,$ \\
{\small{
$p_{0,0} = [0.45,0.50,0.51,0.57,0.63, 0.66,0.69,0.75,0.78,0.87]$ \\ $ p_{1,0}=
[0.45, 0.41,0.40, 0.35, 0.30, 0.28, 0.25, 0.20, 0.15, 0.10].$}}\\
\begin{figure}
  \begin{center}
    \begin{tabular}{cc cc}
    \hspace{-0.5cm}
      \includegraphics[scale=0.155]{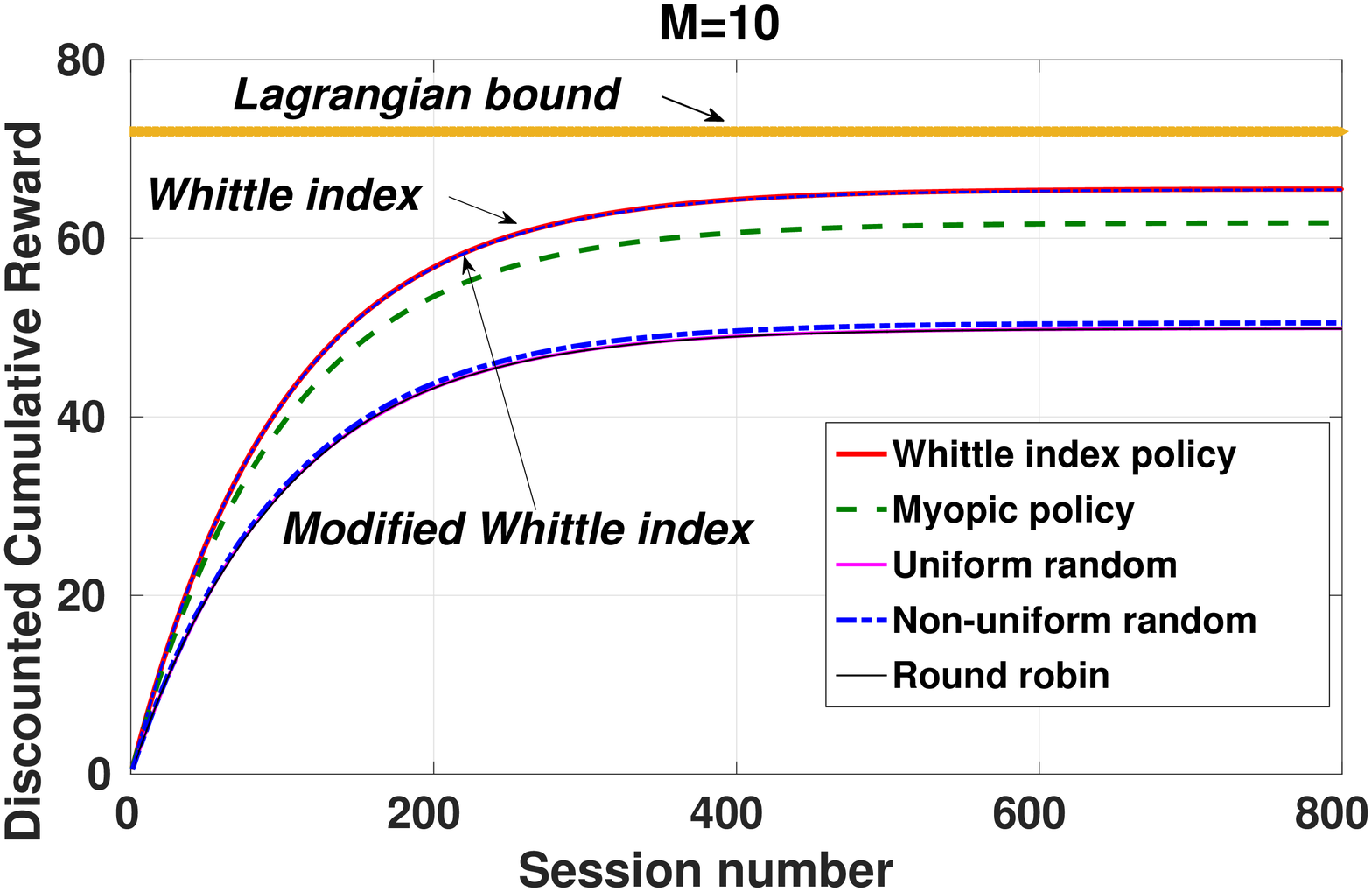}
      & 
      \hspace{-0.8cm}
      \includegraphics[scale=0.155]{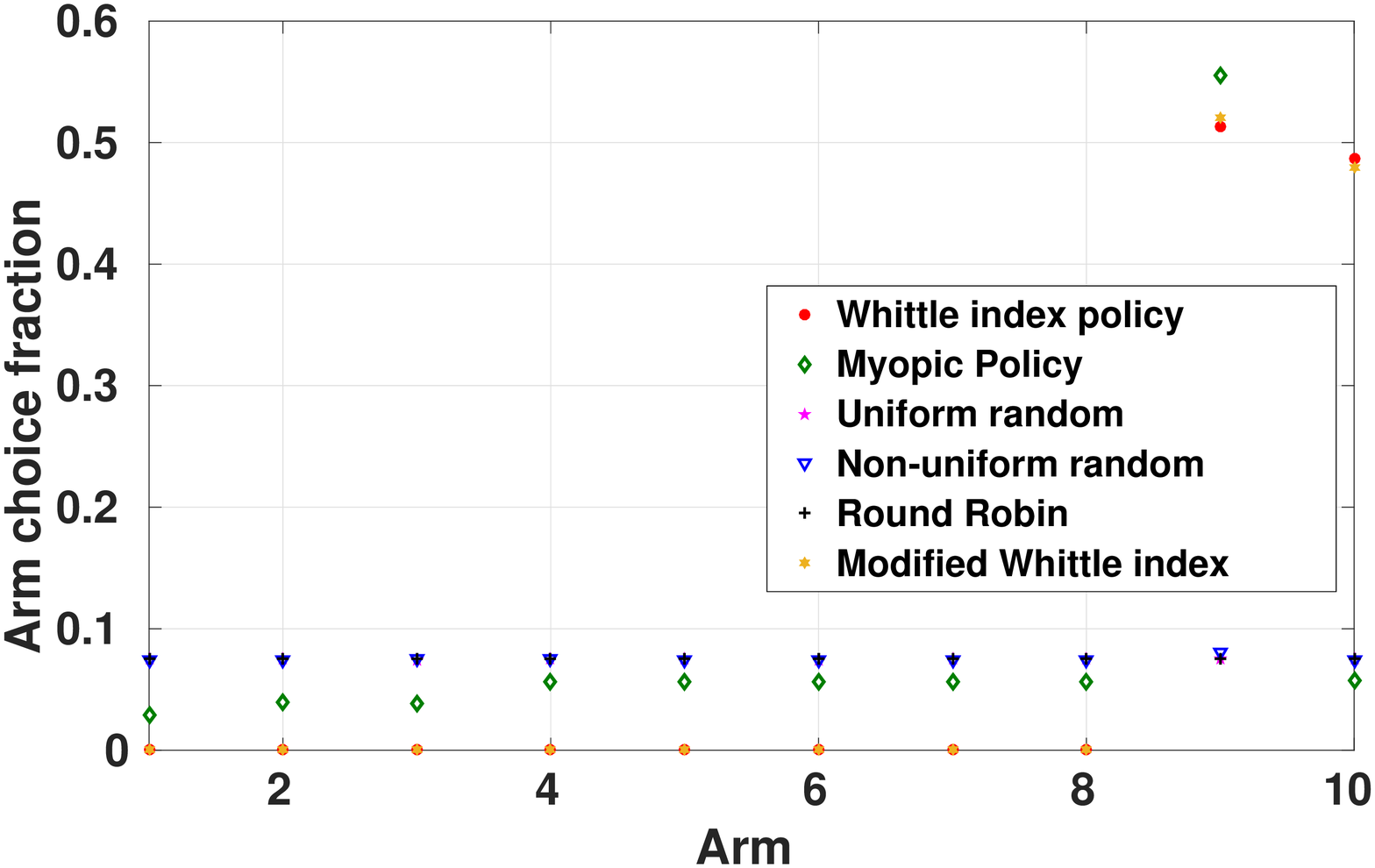} \\  
      a) {\small{ Discounted cumulative reward}}   & \hspace{-0.7cm} b) {\small{ Arm choice fraction }}  
     \end{tabular}
  \end{center} 
  \caption{ Example-1: a) discounted cumulative rewards as function of
    sessions for different policies and b) arm choice fraction for
    each arm with different policies. In this scenario, the modified Whittle index policy performs same as the Whittle index policy} 
  \label{fig:reward-armchoice-Ex1}
\end{figure}
\begin{table}[h]
\centering
\caption{Example-1: Value for different policies; ${\pi}(1)=q$.}
\label{TAB:Ex1-Value-pinit}
\begin{tabular}{ccccccc}
\hline
 $L_b$ &	 WI     & MWI     & MP	 & NUR		& RR	  & Random   \\ \hline \hline
 $72$  & $65.52$ & $65.44$ & $61.73$ & $50.53$  & $49.88$ & $49.91$  \\ \hline
\end{tabular}
\end{table}
In Fig.~\ref{fig:reward-armchoice-Ex1}-a) we can see the discounted
cumulative reward as function of session number, plotted for various policies along with the Lagrangian bound $L_b.$ Table~\ref{TAB:Ex1-Value-pinit} gives the average value generated by various algorithms with stationary probabilities of arms as the initial beliefs. In this case, the discounted cumulative reward obtained by Whittle-index policy (WI) is similar to that of Modified Whittle-index policy (MWI), and higher than
is higher than that of myopic policy (MP) and other policies. 
 In Fig.~\ref{fig:reward-armchoice-Ex1}-b),we can see arm choice fractions of all arms under different policies.  
Notice the tendency of WI and MWI to prefer a smaller subset of arms, $\{9,10\}$  as compared to other policies.
This behavior might be because they account for future rewards through the action value functions. 

\subsection{Example-2 : Positively and negatively correlated arms}
In this example, we consider a more generic parameter set for which no index expressions are available. 
This set consists both of positively and negatively correlated  arms unlike other examples. Whittle indices were computed using Algorithm~\ref{algo:WI}.
We use,{\small{
\begin{eqnarray*}
p_{0,0} = [0.7, 0.6, 0.5, 0.8, 0.6, 0.3, 0.3, 0.2, 0.25, 0.2],\\
 p_{1,0} = [0.2, 0.2, 0.2, 0.3, 0.3, 0.5, 0.6, 0.5, 0.45, 0.7],\\
\rho_0 = [0.2, 0.1, 0.15, 0.3, 0.25, 0.3, 0.2, 0.2, 0.3, 0.1],\\
\rho_1 = [0.8, 0.9, 0.85, 0.9, 0.8,  0.8, 0.8, 0.9, 0.7, 0.9],\\
R_0=[0.1, 0.25, 0.3, 0, 0.15, 0.2, 0.35, 0.25, 0.1, 0.3],\\
R_1 = [  1, 0.85, 0.8, 1.0, 0.95, 0.9, 0.75, 0.85,   1.0, 0.8].
\end{eqnarray*} }}
From Table~\ref{TAB:tccn_ex2}, WI is closer to the Lagrangian bound. Also note the change in the performances of MWI relative to WI and myopic policies, in contrast to Example-1. 
\begin{table}[h]
\centering
\caption{Example 2: Average Value generated by various policies for the same initial belief (randomly chosen),  $K=3.$} \vspace{-0.3 cm}
\label{TAB:tccn_ex2}
\begin{tabular}{ccccccc}
\hline
	${L_b}$     &  WI 	   &  MP     & MWI      & NUR      & RR      & Random  \\ \hline\hline 
	$71.68$	    &  $70.25$ & $68.26$ & $67.87$  & $60.79$  & $60.08$ & $59.68$ 	\\ \hline
\end{tabular}
\end{table} 
\subsection{Example-3 : Effect of multiple state transitions $K$.}
In this example, we study the effect of $K$ on the performance of various policies.
Here, $M =15$ channels were used. For the
  first $10$ channels, $K=20,$ for which $\gamma_2(\pi)
  \approx q.$ For next $5$ channels, $K$ is varied from $1$ to $5.$
  Further, we assumed $R_{m,0} = \rho_{m,0} = 0.$ Other parameters are
  given below.  \\
{\small{
\begin{eqnarray*}
p_{0,0} =[0.50,0.45,0.45,0.78,0.6, 0.6,0.7, 0.7, 0.4,0.45, 0.5,0.6,\\0.7,0.5,0.35],\\
p_{1,0} = [0.41, 0.4,0.35, 0.15, 0.55, 0.5, 0.5, 0.6, 0.3, 0.25,  0.2, 0.2,\\0.2,0.3,0.25],\\
\rho_1 = [0.9, 0.8, 0.8, 0.8, 0.9, 0.9, 0.9, 0.9,  0.8, 0.7,  1.0, 1.0,\\ 1.0 ,1.0 ,1.0 ],\\
R_1 = [0.9, 0.8, 0.8, 0.8, 0.9,  0.9, 0.9, 0.9, 0.8,  0.7,   0.6, 0.7,\\ 0.85, 0.6, 0.7].
\end{eqnarray*} }}
\begin{table}[h]
\centering
\caption{Example 3: Average value generated by various policies. Initial belief - random.} \vspace{-0.3 cm}
\label{TAB:Ex4-Value-effectofK}
\begin{tabular}{ccccccc}
\hline
    $L_b$      &  WI 	   & MWI		&  MP     & NUR      & RR      & Random  \\ \hline\hline 
 	$62.49$    &  $60.48$   & $58.00$	& $55.48$ & $45.35$  & $44.25$ & $44.22$\\ \hline
\end{tabular}
\end{table}
\begin{figure}
  \begin{center}
    \begin{tabular}{cc}    
    \hspace{-0.5cm}
      \includegraphics[scale=0.155]{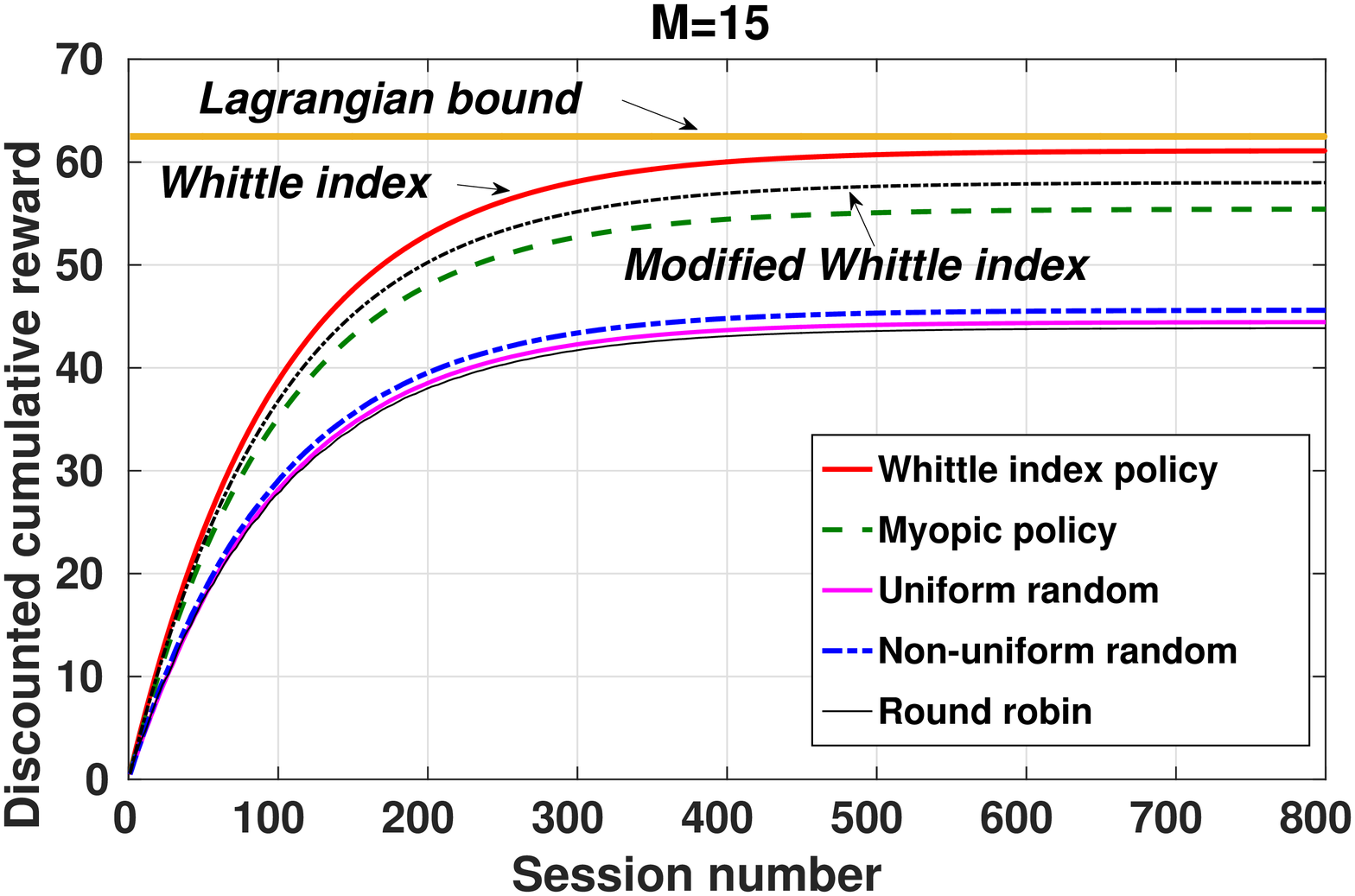}      
      & 
      \hspace{-0.75cm}
      \includegraphics[scale=0.155]{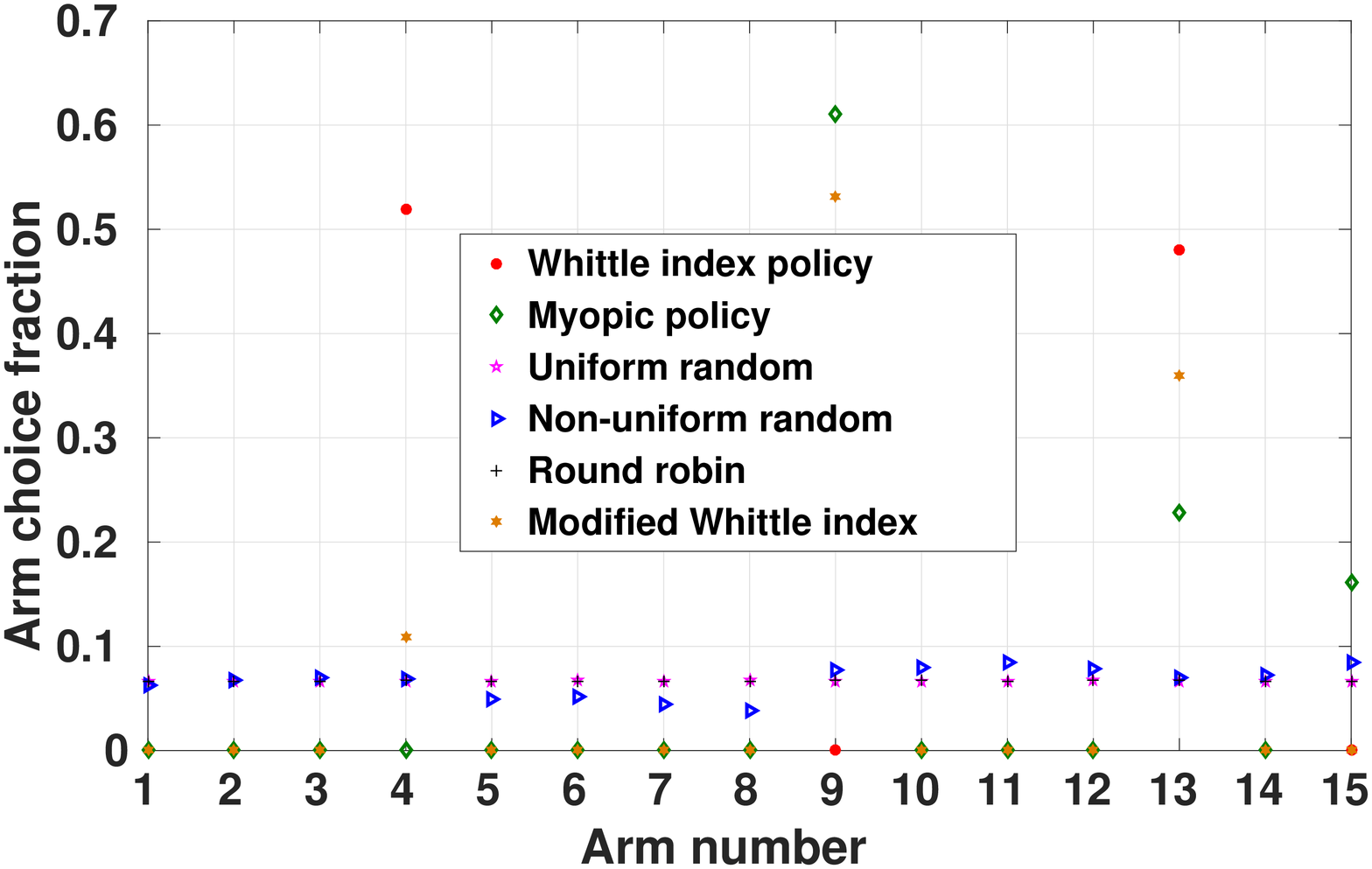} \\  
      a) {\small{ Discounted cumulative reward}}   & \hspace{-0.8cm} b) {\small{ Arm choice fraction }}  
     \end{tabular}
  \end{center} 
 \caption{Example-3: a) The discounted cumulative reward verses
   session number for different policies and b) arm choice fraction
   for each arm with different policies. Notice the index policy almost reaches  up to the Lagrangian upper bound on optimal value. }
  \label{fig:reward-armchoice-Ex4-effectofK}
\end{figure}
Table~\ref{TAB:Ex4-Value-effectofK} and Fig.~\ref{fig:reward-armchoice-Ex4-effectofK} summarize the performance of various policies. In this case, WI almost reaches the Lagrangian bound $L_b.$ The performance of MWI is in between that of WI and myopic policy. 
\subsection{Example-4 : Effect of inaccurate estimates of $K$.}
In this example, we simulate a situation where the decision maker does not know the exact values of $K,$ and proceeds with an inaccurate estimate $K_e$ for all arms. 
We use the same parameters as Example-$3.$ 
For arms $1$ to $10$ the value  $K=20,$ while, for arms $11$ to $15$ respectively have $K$ values $1,2,3,4,5.$ However, as these values are unknown, the value $K_e$ in used in decision making.

Table~\ref{TAB:tccn_ex4} gives the Lagrangian bound along with the average value generated by the index and myopic policies for inaccurate estimates $K_e.$ The ordering on the performances is the same as in Example-$3$.
\begin{table}[h]
\centering
\caption{Example 4: Average Value generated by various policies for inaccurate  estimates $K_e,$ of $K.$} \vspace{-0.3 cm}
\label{TAB:tccn_ex4}
\begin{tabular}{ccccccc}
\hline
 $K_e=$   &  $1$  & $ 2 $  & $3$   & $4$    & $5 $ & $10$  \\ \hline\hline 
 $L_b$  &  $61.07$  & $ 62.12 $  & $ 62.43 $ & $ 62.63 $  & $62.84$  & $63.04$ \\ 
 WI  &  $59.59$  & $ 59.89 $  & $ 60.43 $ & $ 60.56 $  & $60.65$  & $60.74$ \\ 
 MWI &  $56.29$  & $ 56.98 $  & $ 57.38 $ & $ 57.44 $  & $57.47$  & $58.25$ \\
 MP  &  $55.13$  & $ 55.31 $  & $ 55.46 $ & $ 55.47 $  & $55.78$  & $55.73$ \\ \hline
\end{tabular}
\end{table}
Table~\ref{TAB:tccn_ex4} shows the relative gains of the policies when  inaccurate estimates $K_e$ are used for decision making. There are only minor changes in the performances of policies compared to those in Table~\ref{TAB:Ex4-Value-effectofK} where correct $K$ values are known and used for decision making.
\section{Conclusion}
\label{sec:conclusion}
In this work, the problem of restless multi-armed bandits with
cumulative feedback and partially observable states was
formulated. Such bandits are called lazy restless bandits (LRB). 
This model is can be applied for sequential decision making in scenarios where instants of decision making are sparser than instants of system state transition.  LRBs are indexable and the Whittle indices can be computed using a two-timescale stochastic gradient algorithm. 
An upper bound on the optimal value function is provided by the Lagrangian relaxation of the problem. Numerical simulations show that   Whittle-index policy is almost optimal, comes close to the upper bound in some instances. The performance of modified Whittle-index policy is found to vary widely with problem instances. It performs as good as the Whittle-index policy in some instances, worse than myopic policy in some others. 

 It would also be interesting to extend the index policies for RMABs with multiple states that are hidden .
 Other directions for future work include a more detailed study of modified Whittle index and LP based heuristics as alternatives computationally lighter alternatives to Whittle-index policy.
Also, restless bandits with constrained or intermittently available arms would make a useful study.

\bibliographystyle{IEEEbib}
\bibliography{limfee}

\appendix
We provide proofs for the case $\rho_{m,0} < \rho_{m,1}$ and $R_{m,0}<R_{m,1}.$ Recall, we had assumed that the ordering on rewards $R_{m,i}$  is same as the ordering on success probabilities $\rho_{m,i}.$ That is, if $R_{m,0}<R_{m,1},$ then $\rho_{m,0} < \rho_{m,1},$ and vice-versa. 
Note that a result claimed by assuming $R_{m,0}<R_{m,1},$ can also be claimed for the opposing case  $R_{m,0}>R_{m,1},$ using the same proof techniques.
The subscript $m$ is dropped in the analysis of single armed bandits, $\rho_i$ and $R_i$ are used.

  \subsection{Proof of Lemma $1$ - Part $1)$}
  \label{proof:valfunc-convex-pi}
  First we prove convexity of the functions $V_S(\pi),V_{NS}(\pi)$ can be proved using induction.
  It then follows that $V(\pi)$ is convex.
  Let 
  $V_{NS,1}(\pi) =\eta ,$
  $V_{S,1}(\pi) = R_S(\pi)= \pi R_0 + (1-\pi)R_1,$
  $V_1(\pi) = \max\{ V_{S,1}(\pi), V_{NS,1}(\pi) .$ 
  Clearly, $V_{NS,1}(\pi),V_{S,1}(\pi)$  and in turn $V(\pi)$ are convex in $\pi$.
  Assume this convexity claim holds for $V_{NS,n}(\pi),V_{S,n}(\pi).$ 
  Now,
  {\small{
  \begin{align*}
  V_{S,n+1}(\pi) = R_S(\pi) &+ \beta \rho(\pi) V_n(\gamma_1(\pi)) \nonumber\\& +\beta (1-\rho(\pi)) V_n(\gamma_0(\pi)) \nonumber \\
    V_{NS,n+1}(\pi)= \eta &+ \beta V_n(\gamma_2(\pi)) \\
    V_{n+1}(\pi) = \max\{& V_{S,n+1}(\pi), V_{NS,n+1}(\pi) \}.\nonumber
  \end{align*} }}
  Define {\small{
  \begin{align*}
   b_0 :=[(1-\pi)&(1-\rho_{1})p_{10} + \pi(1-\rho_{0}) p_{00},\texttt{ }\\& (1-\pi)(1-\rho_{1})(1-p_{10}) + \pi(1-\rho_{0})(1- p_{00}) ]; \\
   b_1 := [(1-\pi)&\rho_{1}p_{10} + \pi\rho_{0}p_{00},\texttt{ }\\& (1-\pi)\rho_{1}(1-p_{10}) + \pi\rho_{0}(1-p_{00})]; \\
   {\mid\mid b_1 \mid\mid}_1 = \pi\rho_0 &+ (1-\pi)\rho_1  \equiv \rho(\pi);\\
    {\mid\mid b_0 \mid\mid}_1 = 1 - &\pi\rho_0 - (1-\pi)\rho_1 = 1 -\rho(\pi).
  \end{align*}}}
  Now, $V_{S,n+1}(\pi)$ can be rewritten as {\small{
  \begin{dmath*}
   V_{S,n+1}(\pi) = R_S(\pi) + \beta ||b_1||_1 V_n\left(\frac{b_1}{||b_1||_1} \right) + \beta||b_0||_1 V_n\left(\frac{b_0}{||b_0||_1} \right)
  \end{dmath*} }}
  We know that $V_n\left(\pi\right)$ is convex. Using Lemma 2 from \cite{astrom1969}, $||b_1||_1 V_n\left(\frac{b_1}{||b_1||_1} \right)$ is also convex.
  This implies that $V_{S,n+1}$ is a sum of convex functions and hence convex.
  Similarly, $V_{NS,n+1}(\pi) = \eta + \beta V_n(\gamma_2(\pi)).$
  Here, $V_n(\pi)$ is convex and $\gamma_2(\pi)$ is linear. Hence, $V_{NS,n+1}(\pi)$ is convex.
  It follows that $V_{n+1}(\pi)$ is convex. By principle of induction, $V_{S,n}(\pi),V_{NS,n}(\pi)$ and $V_n(\pi)$ are convex for all $n.$ 
  From \cite{bertsekas1995} Chapter 2, as $n\rightarrow \infty,$  $V_{S,n}(\pi) \rightarrow V_S(\pi),$ $V_{NS,n}(\pi) \rightarrow V_{NS}(\pi)$ and $V_{n}(\pi) \rightarrow V(\pi).$ 
  This means that the functions $V_S,V_{NS},V$ are convex in $\pi.$

  \subsection{Proof of Lemma $1$ - Part $2)$}
  \label{proof:valfunc-convex-eta}
  This result too can be claimed using the induction principle. To emphasize that subsidy $\eta$ is a variable, value functions are rewritten as 
  $V_S(\pi,\eta),V_{NS}(\pi)$ and $V(\pi,\eta).$ For a fixed $\pi,$ let 
  $V_{NS,1}(\pi,\eta) = \eta ,$ $
  V_{S,1}(\pi,\eta) = R_S(\pi)= \pi R_0 + (1-\pi)R_1,$ and $
  V_1(\pi,\eta) = \max\{ R_S(\pi),\eta \}.$
  Clearly, all the above functions are convex and non-decreasing in $\eta$. 
  Now suppose $V_{S,n}(\pi,\eta),V_{NS,n}(\pi,\eta)$ and in turn $V_{n}(\pi,\eta)$ are convex.{\small{
  \begin{align*}
  V_{S,n+1}(\pi,\eta) = R_S(\pi) &+ \beta \rho(\pi) V_n(\gamma_1(\pi),\eta) \\ &+ \beta(1-\rho(\pi)) V_n(\gamma_0(\pi),\eta)  \\
   V_{NS,n+1}(\pi,\eta)= \eta &+ \beta V_n(\gamma_2(\pi),\eta)  \\
    V_{n+1}(\pi,\eta) = & \max\{ V_{S,n+1}(\pi), V_{NS,n+1}(\pi) \}.
   \end{align*} }}
  Here, $V_{NS,n+1}(\pi,\eta)$ is non-decreasing convex in $\eta$ because it is a sum of two non-decreasing convex functions in $\eta$.
  Further, $V_{S,n+1}(\pi,\eta)$ is sum of a constant function and a convex combination of two non-decreasing convex functions; hence it is convex non-decreasing.
  By induction $V_{S,n},V_{NS,n}$ and $V_{n}$ are non-decreasing convex for any $n\geq 1.$ 
  As in  part 1) of this lemma, as $n\rightarrow \infty,$  $V_{S,n}(\pi,\eta) \rightarrow V_S(\pi,\eta),$ $V_{NS,n}(\pi,\eta) \rightarrow V_{NS}(\pi,\eta)$ 
  and $V_{n}(\pi,\eta) \rightarrow V(\pi,\eta).$ 
  This means that the functions $V_S,V_{NS},V$ are convex and non-decreasing in $\eta$ for fixed $\pi.$
  \qed

\subsection{Proof of Lemma $3$}
\label{proof:valfunc-dec-for-positivecorr}
 The proof is done by the principle of induction.  Assume that $V_n(\pi)$ is
 non increasing in $\pi.$ Let $\pi'> \pi$ and consider playing the
 arm is optimal. Then
{\small{
\begin{equation*}
V_{n+1}(\pi) = R_S(\pi) + \beta \left[\rho(\pi) V_n(\gamma_1(\pi)) + 
  (1-\rho(\pi)) V_n(\gamma_0(\pi)) \right]
\end{equation*}
}}
Here $R_S(\pi) = \pi R_0 + (1-\pi) R_1.$ Note that $R_S(\pi)$ is
decreasing in $\pi,$ i.e. $R_S(\pi') < R_S(\pi)$ whenever $\pi'> \pi.$
Hence we get 
{\small{
\begin{equation}
\label{vdec}
V_{n+1}(\pi) \geq R_S(\pi') + \beta \left[\rho(\pi) V_n(\gamma_1(\pi)) + 
  (1-\rho(\pi)) V_n(\gamma_0(\pi)) \right].
\end{equation}
}}
 From our assumptions $p_{00}>p_{10}$ and $\rho_1 > \rho_0$, we get
 a stochastic ordering $(\leqslant_s)$ on observation probabilities, i.e., $[1-\rho(\pi'),\rho(\pi')]^T \leqslant_{s} [1-\rho(\pi),\rho(\pi)]^T.$ 
   Also, $\gamma_0(\pi) \geq \gamma_1(\pi);$ and as $V_n(\pi)$ is decreasing in $\pi,$  $V_n(\gamma_0(\pi))\leq V_n(\gamma_1(\pi)).$ Then, using a property of stochastic ordering \cite[Lemma 1.1]{Lovejoy87} along with \eqref{vdec}, we obtain 
{\small{
\begin{equation*}
V_{n+1}(\pi) \geq R_S(\pi') + \beta \left[\rho(\pi') V_n(\gamma_1(\pi)) + 
  (1-\rho(\pi')) V_n(\gamma_0(\pi)) \right].
\end{equation*}
}}
Now that $\gamma_0,\gamma_1$ are increasing in $\pi$ and $V_n$ is decreasing in
$\pi$, we have   
{\small{
\begin{dmath*}
V_{n+1}(\pi) \geq R_S(\pi') + \beta \left[\rho(\pi') V_n(\gamma_1(\pi')) + 
  (1-\rho(\pi')) V_n(\gamma_0(\pi')) \right] \\ \geq V_{n+1}(\pi').
\end{dmath*} }}
%
%
This is true for every $n.$ From \cite{Ross71}, we know $V_n(\pi)
\rightarrow V(\pi)$ as $n\rightarrow \infty$. Thus $V(\pi)$ is decreasing in $\pi.$
Similarly, when not playing the arm is optimal, it is clear that $V_{n+1}(\pi) = \eta + V_n(\gamma_2(\pi)) \geq V_{n+1}(\pi'),$ for $\pi'>\pi,$ as $\gamma_2$ is increasing in $\pi$ for positively correlated arms. Likewise, the same stochastic ordering argument for showing $V_S$ is decreasing in $\pi.$\qed 
\subsection{Proof of Lemma $4$ }
\label{proof:dpi-decreasing-positivecorr}
\subsubsection{Part 1) - For large $K$}
\label{proof:dpi-decreasing-largeK}
Let $d(\pi):= V_S(\pi)-V_{NS}(\pi).$ We want to prove that $d(\pi)$
decreasing in $\pi.$ This implies that we need to show $V_S(\pi)-V_{NS}(\pi) < V_S(\pi')-V_{NS}(\pi'), \text{ whenever } \pi>\pi' .$ That is to show $V_S(\pi)-V_{S}(\pi') < V_{NS}(\pi)-V_{NS}(\pi').$

In our setting $V_{NS}(\pi)-V_{NS}(\pi') = 0$ whenever $\gamma_2(\pi)
= q$ and this is true for large values of $k.$ 
We know for positive correlated arms,
$V_S(\pi)-V_{S}(\pi') < 0,$ as $V_S$ is decreasing in $\pi.$  
Hence, the claim follows.
\qed 
\subsubsection{Parts 2),3) - for any $K>1$}
\label{proof:dp-dec-arbitrary-K}
To prove $V_S - V_{NS}$ is decreasing with $\pi,$ we need the following result.
\addtocounter{lemma}{8}
\begin{lemma}
\label{lemma:lipschitz}
The functions $\left|\frac{\partial V(\pi)}{\partial \pi} \right|, \left|\frac{\partial V_S(\pi)}{\partial \pi} \right|$ and $\left|\frac{\partial V_{NS}(\pi)}{\partial \pi} \right|$ $\leq \kappa c (\rho_1 - \rho_0),$ when, $\beta <\frac{1+b}{4}$ or $0<|p_{0,0}-p_{1,0}|<\frac{1+b}{4}.$ Here $\kappa = \frac{1}{1-\beta|p_{0,0}-p_{1,0}|},$ $b = \min\left\lbrace 1,\frac{R_1-R_0}{\rho_1-\rho_0} \right\rbrace$ and $c = \max\left\lbrace 1,\frac{R_1-R_0}{\rho_1-\rho_0} \right\rbrace$
\end{lemma}
\begin{proof}
We prove this by induction. 
We provide the proof for the case $p_{0,0}>p_{1,0},$ \textit{i.e.,} positively correlated arms. The same procedure also works for $p_{0,0}<p_{1,0},$ \textit{i.e.} negatively correlated arms.  Also, notice that $\kappa \geq 1.$
\begin{enumerate}
\item $V_{S,1} = R_S(\pi),$ $V_{NS,1} = \eta$ and so, $V_1(\pi) = \max\{R_S(\pi),\eta\}.$ Clearly, as all the functions are convex,
$\mid\frac{\partial V_1(\pi)}{\partial \pi} \mid \leq \kappa c(\rho_1 - \rho_0). $ 
q\item Assume $\left| \frac{\partial V_n(\pi)}{\partial \pi} \right| < \kappa c(\rho_1 - \rho_0) .$
\item Now,
\begin{eqnarray*}
  V_{S,n+1}(\pi) = R_S(\pi) + \beta \rho(\pi) V_{n}(\gamma_1(\pi)) \\ + 
  \beta(1-\rho(\pi)) V_{n}(\gamma_0(\pi)) \\
  V_{NS,n+1}(\pi)= \eta + \beta V_{n}(\gamma_2(\pi))  \\
  V_{n+1}(\pi) = \max\{ V_{S,n+1}(\pi), V_{NS,n+1}(\pi) \}.\nonumber
\end{eqnarray*} 
{\small{
\begin{eqnarray}
\begin{split}
\frac{\partial V_{S,n+1}(\pi)}{\partial \pi} = &(R_0 - R_1) + \\& \beta(\rho_1-\rho_0)\left[ V_n(\gamma_0(\pi)) - V_n(\gamma_1(\pi)) \right]  + \\& \beta \rho(\pi)\frac{\partial V_{n}(\gamma_1(\pi))}{\partial \pi} \gamma_{1}^'(\pi)  + \\&
\beta(1-\rho(\pi))\frac{\partial V_{n}(\gamma_0(\pi))}{\partial \pi} \gamma_{0}^'(\pi).
\end{split}
\end{eqnarray}
}}
Substituting $\gamma_{1}^'(\pi) = \frac{\rho_1 \rho_0(p_{0,0}-p_{1,0})}{(\rho(\pi)^2)}$ and  
$\gamma_{0}^'(\pi) = \frac{(1-\rho_1)(1-\rho_0)(p_{0,0}-p_{1,0})}{(1-\rho(\pi))^2},$
{\small{
\begin{eqnarray*}
\begin{split}
\frac{\partial V_{S,n+1}(\pi)}{\partial \pi} &= (R_0 - R_1) +\\& \beta(\rho_1-\rho_0)\left[ V_n(\gamma_0(\pi)) - V_n(\gamma_1(\pi)) \right]  +\\& \beta\frac{\partial V_{n}(\gamma_1(\pi))}{\partial \pi} \frac{\rho_1 \rho_0(p_{0,0} -p_{1,0})}{(\rho(\pi))}  + \\&
\beta\frac{\partial V_{n}(\gamma_0(\pi))}{\partial \pi} \frac{(1-\rho_1)(1-\rho_0)(p_{0,0}-p_{1,0})}{(1-\rho(\pi))}.
\end{split}
\end{eqnarray*} }}
\item Bound: We know that $\rho(\pi) \in [\rho_0,\rho_1]$ and $ 1-\rho(\pi) \in [1-\rho_1,1-\rho_0].$ Substituting in above equation, we have 
{\small{
\begin{dmath*}
 \frac{\partial V_{S,n+1}(\pi)}{\partial \pi} \leq (R_0 - R_1) +  \beta(\rho_1-\rho_0)\left[ V_n(\gamma_0(\pi)) - V_n(\gamma_1(\pi)) \right] + \beta\frac{\partial V_{n}(\gamma_1(\pi))}{\partial \pi} 
 {\rho_1 (p_{0,0}-p_{1,0})} + \beta\frac{\partial V_{n}(\gamma_0(\pi))}{\partial \pi} {(1-\rho_0)(p_{0,0}-p_{1,0})}.
\end{dmath*}
 }}
 From the assumption in Step 2), we can conclude that ${V_n(\gamma_0(\pi))-V_n(\gamma_1(\pi))}\leq \kappa c(\rho_1-\rho_0)(\gamma_0(\pi) - \gamma_1(\pi)).$ Further we have $\gamma_0(\pi) - \gamma_1(\pi) \leq p_{0,0}-p_{1,0}.$ Using this, we have {\small{
\begin{dmath*}
\frac{\partial V_{S,n+1}(\pi)}{\partial \pi} \leq (R_0 - R_1) + {\beta}(\rho_1-\rho_0)^2 \kappa c(p_{0,0}-p_{1,0}) + \beta \kappa c(p_{0,0}-p_{1,0})\rho_1(\rho_1-\rho_0) + \beta\kappa c(p_{0,0}-p_{1,0})(1-\rho_0)(\rho_1-\rho_0)
\\ \leq (R_0 - R_1) + {(\rho_1-\rho_0)\{ \beta \kappa c(p_{0,0}-p_{1,0})[1 + 2(\rho_1-\rho_0)] \}} \\ {\leq (R_0 - R_1) + {(\rho_1-\rho_0)\{ 3\beta \kappa c(p_{0,0}-p_{1,0}) \}}} 
\end{dmath*} }} 
Rewriting the R.H.S. of the above inequality, we obtain
\begin{equation}
\label{lips-eqn}
\frac{\partial V_{S,n+1}(\pi)}{\partial \pi} \leq \kappa c(\rho_1-\rho_0)\{-b+4\beta(p_{0,0}-p_{1,0})\}
\end{equation} where, $b = \min\left\lbrace 1,\frac{R_1-R_0}{\rho_1-\rho_0} \right\rbrace$ and $c = \max\left\lbrace{ 1,\frac{R_1-R_0}{\rho_1-\rho_0} }\right\rbrace.$
If $\beta <\frac{(1+b)}{4}$ or $0<p_{0,0}-p_{1,0}<\frac{(1+b)}{4},$ then, $|-b + 4\beta(p_{0,0}-p_{1,0})| \leq 1$. Now, it follows that
$ \left| \frac{\partial V_{S,n+1}(\pi)}{\partial \pi}\right| \leq \kappa c(\rho_1-\rho_0).$
Hence, by principle of induction, the claim is true for all $n>0$. By the property of the value function that $\lim_{n\rightarrow \infty} V_{S,n}(\pi) = V_S(\pi),$ it follows that $\left|\frac{\partial V_S(\pi)}{\partial \pi}\right| < \kappa c(\rho_1-\rho_0).$  
\item Similarly, {{
\begin{eqnarray}
\begin{split}
\frac{\partial V_{NS,n+1}(\pi)}{\partial\pi} &= \beta \frac{\partial V_{NS,n}(\gamma_2(\pi))}{\partial(\gamma_2(\pi))} \gamma_2^'(\pi)\nonumber\\& \leq \beta \kappa c (\rho_1-\rho_0) (p_{0,0}-p_{1,0})^K \\& \leq \kappa c(\rho_1-\rho_0). 
\end{split}
\label{lips-eqn-ns}
\end{eqnarray} }}
Hence, by principle of induction, the claim is true for all $n>0$. By the property of value function that $\lim_{n\rightarrow \infty} V_{NS,n}(\pi) = V_{NS}(\pi),$ it follows that $\left|\frac{\partial V_{NS}(\pi)}{\partial\pi}\right| < \kappa c(\rho_1-\rho_0).$  \qed
\end{enumerate}
Now, consider $d(\pi) = V_S(\pi)-V_{NS}(\pi).$ It is enough to show $\frac{\partial V(\pi)}{\partial\pi} < 0.$ {\small{
\begin{dmath*}
\frac{\partial d(\pi)}{\partial\pi} = \frac{\partial V_S(\pi)}{\partial \pi} - \frac{\partial V_{NS}(\pi)}{\partial \pi}
\end{dmath*}}}
From \eqref{lips-eqn}, we have {\small{
\begin{dmath*}
\frac{\partial V_S(\pi)}{\partial\pi} \leq \kappa c{(\rho_1-\rho_0)\{-b + 4\beta(p_{0,0}-p_{1,0}) \}}
\end{dmath*} 
\begin{dmath*}
\frac{\partial V_{NS}(\pi)}{\partial\pi} \geq -\beta\kappa c(\rho_1-\rho_0)|p_{0,0}-p_{1,0}|^K.
\end{dmath*}
\begin{dmath*}
\frac{\partial d(\pi)}{\partial\pi} \leq \kappa c{(\rho_1-\rho_0)\{-b + 4\beta(p_{0,0}-p_{1,0}) \}}  + 
\beta\kappa c{(\rho_1-\rho_0)|p_{0,0}-p_{1,0}|^K}  \leq \kappa c{(\rho_1-\rho_0)\{-b + 5\beta(p_{0,0}-p_{1,0}) \}}
\end{dmath*} }}
In the R.H.S of above inequality, $\{-b + 5\beta(p_{0,0}-p_{1,0}) \} <0$ when, $0<p_{0,0}-p_{1,0}<\frac{b}{5}$ or $\beta<\frac{b}{5}$. Also, $\frac{b}{5} < \frac{1+b}{4},\forall b>0.$ Hence, under these conditions $V_S - V_{NS}$ is decreasing in $\pi.$  \qed
\end{proof} 
\subsection{Proof of Theorem $2$}
\label{proof:thm-indexability-general}
Using induction technique, one can obtain the following inequalities.{\small{
\begin{equation*}
\bigg\vert \frac{\partial V(\pi, \eta)}{\partial \eta } \bigg \vert, 
\bigg\vert \frac{\partial V_S(\pi, \eta)}{\partial \eta } \bigg \vert,
\bigg\vert \frac{\partial V_{NS}(\pi, \eta)}{\partial \eta } \bigg \vert
\leq \frac{1}{1-\beta}
\end{equation*} }}
Also, 
{\small{
\begin{eqnarray*}
\frac{\partial V_S(\pi, \eta)}{\partial \eta} 
= \beta \left[ \rho(\pi) \frac{\partial V(\gamma_1(\pi),\eta)}{\partial \eta } +\right.\\ \left.
  (1-\rho(\pi)) \frac{\partial V(\gamma_0(\pi),\eta)}{\partial \eta}\right] 
\end{eqnarray*} 
\begin{eqnarray*}
\frac{\partial V_{NS}(\pi, \eta)}{\partial \eta } = 1 + \beta \frac{\partial V(q,\eta)}{ \partial \eta }.
\end{eqnarray*}
}} 
Now taking differences {\small{
\begin{eqnarray*}
\begin{split}
\frac{\partial V_{NS}(\pi, \eta)}{\partial \eta } &- \frac{\partial V_S(\pi, \eta)}{\partial \eta } = 1 + \beta
\frac{\partial V(q,\eta )}{ \partial \eta } -\\& 
\beta \left[
  \rho(\pi) \frac{\partial V(\gamma_1(\pi), \eta)}{\partial \eta } + 
 \right.  \left.
  (1-\rho(\pi)) \frac{\partial V(\gamma_0(\pi), \eta)}{\partial \eta
  }\right]
\end{split}
\end{eqnarray*} }}
From Lemma 7, we require the above
difference to be non-negative at $\pi_T(\eta)$. This reduces to the
following expression.
{\small{ 
\begin{dmath}
 \left[ \rho(\pi) \frac{\partial
    V(\gamma1(\pi), \eta)}{\partial \eta } + 
  (1-\rho(\pi)) \frac{\partial V(\gamma_0(\pi), \eta)}{\partial \eta
  }\right]- \frac{\partial V(q,\eta )}{ \partial \eta }  <   \frac{1}{ \beta}.
\label{eqn:nec-cond}
\end{dmath}
}}
Note that we can provide upper bound on LHS of above expression and it
is upper bounded by $2/(1-\beta).$ If $\beta < 1/3,$
Eqn.~\eqref{eqn:nec-cond} is satisfied. $\pi_T(\eta)$ is decreasing in
$\eta.$ Thus indexability claim follows. \qed 
\subsection{Index computation for arbitrary $K,$ $p_{0,0} > p_{1,0},$ $\rho_0=0 ,\rho_1=1.$}
\label{app:index-g2_neq_q}
\begin{enumerate}
 \item For $\pi \in A_1,$ $
 V_S(\pi) = R_S(\pi) + \beta(1-\pi)V(p_{1,0}) + \beta\pi V(p_{0,0}) .$
 Assuming the threshold is at $\pi,$ we have $V(p_{1,0}) = V_{NS}(p_{1,0})$ and $V(p_{0,0}) = V_{NS}(p_{0,0}).$ Further, for any $\pi < q,$ $\gamma_2(\pi)>\pi.${\small{
 \begin{eqnarray*}
  V_{NS}(\pi) &=& \eta + \beta V(\gamma_2(\pi)) 
    = \eta + \beta V_{NS}(\gamma_2(\pi)) \\
    &=& \eta + \beta (\eta + \beta V_{NS}(\gamma^2_2(\pi))) \\
    &=& \eta(1+\beta + \beta^2 + ...+ \beta^{t-1}) + \beta^t V_{NS}(\gamma^t_2(\pi))).
 \end{eqnarray*}}}
 As $t\rightarrow \infty,$ $\gamma^t_2(\pi)\rightarrow q.$ And it follows that $V_{NS}(\pi) = \frac{\eta}{1-\beta}.$ Further, putting above equations together we have $  V_S(\pi) = R_S(\pi) + \beta \frac{\eta}{1-\beta}.$
 Because $\pi$ is the threshold, the subsidy $\eta$ required such that $V_S(\pi)=V_{NS}(\pi)$ is the index $W(\pi).$ Hence, we get
 $W(\pi) = R_S(\pi).$
 \item For $\pi \in A_2,$
 Assume that the threshold is at $\pi.$ Hence, $V(p_{1,0}) = V_{S}(p_{1,0})$ and $V(p_{0,0}) = V_{NS}(p_{0,0}).$  And $   V_S(\pi) = R_S(\pi) + \beta(1-\pi)V_S(p_{1,0}) + \beta\pi V_{NS}(p_{0,0}) .$ Further, {\small{
 \begin{eqnarray*}
 \begin{split}
  V_S(p_{1,0}) = R_S(p_{1,0}) &+ \beta(1-p_{1,0})V_S(p_{1,0}) \\&+ \beta p_{1,0} V_{NS}(p_{0,0})\\
  \end{split}
  \end{eqnarray*}
\begin{eqnarray*}
 \begin{split}
  V_S(p_{1,0}) &= \frac{R_S(p_{1,0})}{1-\beta(1-p_{1,0})} + \frac{\beta p_{1,0}}{1-\beta(1-p_{1,0})} V_{NS}(p_{0,0})  
  \\&\equiv  a + b V_{NS}(p_{0,0}).
 \end{split} 
 \end{eqnarray*} }}
 Now, just as in the previous interval, it can easily be shown for any $\pi \in A_2,$ that,
 $ V_{NS}(\pi) = \frac{\eta}{1-\beta} .$ Now,  $ V_S(p_{1,0}) = a + \frac{b\eta}{1-\beta.}$
 Then, we have{\small{
 \begin{dmath*}
  V_S(\pi) = R_S(\pi) + \beta (1-\pi)\left(a + \frac{b\eta}{1-\beta} \right) + \beta \frac{\eta}{1-\beta}.
 \end{dmath*} }}
 Equating $V_S$ and $V_{NS},$ we get 
{\small{ 
 \begin{equation}
  W(\pi) = \frac{(1-\beta)[R_S(\pi)+\beta(1-\pi)a]}{1-\beta[\pi + (1-\pi)b]}.
 \end{equation}
}}
 \item For $\pi \in A_3,$ {\small{
  \begin{eqnarray*}
  \begin{split}
   V_S(\pi) = R_S(\pi) &+ \beta(1-\pi)V_S(p_{1,0}) \\&+ \beta\pi V_{NS}(p_{0,0}) \\
  V_S(p_{1,0}) = R_S(p_{1,0}) &+ \beta(1-p_{1,0})V_S(p_{1,0}) \\&+ \beta p_{1,0} V_{NS}(p_{0,0}).
  \end{split}
 \end{eqnarray*} }}
 Hence,{\small{
 \begin{eqnarray*}
 \begin{split}
   V_S(p_{1,0}) &= \frac{R_S(p_{1,0})}{1-\beta(1-p_{1,0})} + \frac{\beta p_{1,0}}{1-\beta(1-p_{1,0})} V_{NS}(p_{0,0}) \\&\equiv  a + b V_{NS}(p_{0,0}).
 \end{split} 
 \end{eqnarray*} }}
 Further,{\small{
\begin{eqnarray}
\begin{split}
  V_{NS}(p_{0,0}) &= \eta + \beta V(\gamma_2(p_{0,0})) \\&= \eta(1+\beta+ ... +\beta^{t-1}) + \beta^t V(\gamma^t_2(p_{0,0})) \nonumber \\ &= \eta\frac{1-\beta^t}{1-\beta} + \beta^t V(\gamma^t_2(p_{0,0})) \\&\equiv \eta e +  \beta^t V(\gamma^t_2(p_{0,0}))
  \end{split}
 \label{II}
 \end{eqnarray}}} where, $t = inf\{l\geq 1 : \gamma_2^l(p_{0,0})\leq  \pi\}.$ Then, $V(\gamma^t_2(p_{0,0})) = V_S(\gamma^t_2(p_{0,0})).$
 Now, {\small{
 \begin{eqnarray*}
 \begin{split}
 \label{(III)}
  V_S(\gamma^t_2(p_{0,0})) = R_S(\gamma^t_2(p_{0,0})) &+ \beta(1-\gamma^t_2(p_{0,0}))V_S(p_{1,0})\\& + \beta\gamma^t_2(p_{0,0})V_{NS}(p_{0,0})
  \end{split} 
 \end{eqnarray*} }}
Let $\gamma^t_2(p_{0,0})\equiv g_2. $ Now,
{\small{
 \begin{eqnarray*}
 \begin{split}
  V_{NS}(p_{0,0}) = \eta e + \beta^t[R_S(g_2) &+ \beta(1-g_2)V_S(p_{1,0}) \\&+ \beta g_2 V_{NS}(p_{0,0})] \\
  = \frac{\eta e}{1-\beta^{t+1}g_2} &+ \frac{\beta^t R_S(g_2)}{1-\beta^{t+1}g_2} 
    \\&+ \frac{\beta^{t+1}(1-g_2)}{1-\beta^{t+1}g_2 }V_S(p_{1,0})\\ \equiv  \eta f + a_1 +b_1V_S(p_{1,0})
    \end{split}
 \end{eqnarray*} }}
 Using above equations, {\small{
 \begin{eqnarray*}
  V_{NS}(p_{0,0}) &=& \eta f + a_1 + b_1 V_S(p_{1,0}) \\ &=& \eta f + a_1 + b_1[a+ b V_{NS}(p_{0,0})] \\
 V_{NS}(p_{0,0}) &=& \frac{a_1+b_1a}{1-bb_1} + \frac{\eta f}{1-bb_1} {\equiv  \eta c + d.} \\
  V_S(p_{1,0}) &=& a + b(\eta c + d) =  \eta bc + a + bd.   
 \end{eqnarray*}
 \begin{eqnarray*}
 \begin{split}
   V_S(\pi) = R_S(\pi) &+ \beta(1-\pi)[\eta bc + a + c\pi] + \beta \pi[\eta c + d]\\ =  R_S(\pi) &+ \beta[(1-\pi)(a+bd)+\pi d] \\&+ \eta \beta c[\pi +(1-\pi)b] 
  \\& \equiv   D(\pi) + \eta B(\pi).
  \end{split}
 \end{eqnarray*} }}
 Using above equations, it follows that {\small{
 \begin{eqnarray*}
  V_{NS}(\pi) &=& \eta + \beta V_S(\gamma_2(\pi))\\
   		   &=&\eta + \beta [\eta B(\gamma_2(\pi)) + D(\gamma_2(\pi))] 
 \end{eqnarray*} }}
 Equating $V_S(\pi)$ and $V_{NS}(\pi)$, we get, 
{\small{
\begin{equation*}
  W(\pi) = \frac{D(\pi)-D(\gamma_2(\pi))}{1+\beta B(\gamma_2(\pi))-B(\pi)}.
\end{equation*} }}
 \item For $\pi \in A_4,$ $
   V_S(\pi)= R_S(\pi) + \beta(1-\pi)V_S(p_{1,0}) + \beta\pi V_{S}(p_{0,0}). $
 We need to compute 
  $V_S(p_{1,0})$ and $V_S(p_{0,0}).$ 
 In this case the optimal action for the $\pi$ is not sample the arm
  once and later sample the arm always. Similarly if the initial
  action is to sample the arm and later the optimal action is to
  sample the arm always. This behavior can be observed from the
  operation $\gamma_0(\pi),$ which is smaller than $p_{0,0}.$ Then, one
  can easily show by induction that,  $V_S(\pi)$ is
  linear in $\pi$ 
  with slope $m$ and intercept $c_1$ as mentioned earlier.  That is,
  $V_S(p_{1,0}) = m p_{1,0} + c_1 $ and $V_S(q) = m q + c_1.$
%
%
 Now, 
 {\small{
 \begin{eqnarray*}
 \begin{split}
  V_S(\pi) = R_S(\pi) &+ \beta(1-\pi)(mp_{1,0}+c_1) \\&+ \beta\pi(mp_{0,0}+c_1) 
   \\& =  m\pi + c_1
 \end{split}
 \end{eqnarray*}  }}
 Simplifying above equations, we have $m = \frac{R_0-R_1}{1-\beta(p_{0,0}-p_{1,0})}, c_1 = \frac{R_1+m\beta p_{1,0}}{1-\beta} .$\\
 Further, $V_{NS}(\pi) = \eta + \beta V_S(\gamma_2(\pi)) 
  			  = \eta + \beta [ m\gamma_2(\pi)+c_1 ].$
 Equating $V_S(\pi)$ and $V_{NS}(\pi),$ we have 
\begin{equation*} 
   W(\pi) = m\pi + c_1 - \beta(m\gamma_2(\pi)+c_1 ) .
\end{equation*}
\end{enumerate} 

\subsection{Discussion}
\subsubsection{Markovian representation of fading channels}
The proposed model is aimed at fading channels that satisfy Markovian assumption. Commonly
used channel fading models such as Rayleigh fading and the resultant
exponential SNR distribution, can be represented as finite state
Markov chains (FSMC) \cite{Wang95,Zhang99,Park09}. One method involved
is to partition the fading coefficients value range such that, the
duration spent in each state is the same, say $\tau.$ This $\tau$
depends on parameters of the fading distribution of the
channel. Suppose that, each of $M$ different channels is represented
using an $n$-state FSMC. If the number of states $n$ and the length of
interval $\tau$ is fixed, then the FSMC representation of
heterogeneous channels may threaten the validity of the conventional
RMAB model assumption---one state transition per decision
interval. Our proposed model remedies this problem by allowing
multiple state transitions ($K$) per decision interval and also, arms
with different values of $K.$ The authors plan to take up the case of random $K$ drawn from a known distribution in future.
\subsubsection{Whittle-index policy gain in moderately sized systems}
In our numerical study, we presented numerical examples with moderate
size systems in terms of the number of arms, say $M=10,15$. From these
examples, we observed that index policy performs better than other
policies; but, it difficult to say about the optimality of the index
policy. However, the Whittle-index policy has been shown to be
asymptotically optimal for large systems,see
\cite{Ouyang12,Larranaga16}, where both the number of arms $N$ and the
arms to be played per decision interval are large. The analysis is
usually done by using fluid approximation technique. The analysis
provided in Section III is valid for systems of all sizes
with $M\geq 1$ and $K\geq 1.$
\subsubsection{Approximation causes loss of history}
When $K$ is large, we made
use of the approximation $\gamma_{2,m}(\pi) \approx q_m,$ for all $m.$
Here, the belief update $\gamma_{2,m}(\pi)$ when not playing the arm,
becomes independent of prior belief $\pi.$ That is, the history of
actions and observations described by sufficiently by statistic $\pi,$
is forgotten. Hence, we say that approximation causes loss of history.  
It is possible that,
this loss of history due to approximation is detrimental to gains obtained by
Whittle-index policy. The index policy tries to maximize conditional expectation of long term rewards, loses valuable historical information along the way as a
result of this memory loss.  Hence, when approximation is used, its
performance might be lowered and become comparable to myopic policy.
%
\subsubsection{Trade off between information gain and computational load}
Notice that, for the case $\rho_0 = 0$ and $\rho_1 = 1,$ the state of an arm at the beginning (first slot) of a session when it is played, becomes exactly observable at the end of the session through the feedback ACK/NACK. Hence, we could obtain a closed form expression for Whittle-index. Similarly, for $\rho_0=0$ and $0<\rho_1 <1,$ the state of arm  at the beginning of the session becomes known when ACK is observed at the end, and not otherwise.  Even then, we were able to compute
the index expression except for one belief region. This suggests that, the less
information gained from playing an arm, the more difficulty in
obtaining closed form expression for index. We devised an index
computation Algorithm 1 for the general case,
\textit{i.e.,} $0<\rho_0, \rho_1<1.$ Index computation algorithms
based on stochastic approximation schemes are often computationally
taxing because index needs to be computed offline and stored. Also,
the accuracy of index computation algorithm depends on tolerance level
$h.$ A smaller tolerance level $h$ requires larger computation time
and vice-versa. Thus, there is often a trade off between the
information gain and computational load.  
\subsubsection{Complexity of online implementation}
Two cases arise during online implementation of the model. The first
is when closed form index expressions are available, where the
decision maker for an $M$-armed RMAB would require to make
$\mathcal{O}(M)$ computations every session. The second case is when
closed form index expressions are not available and the decision maker
searches for the appropriate index value from a database of stored
values from offline computation over a grid of belief points.
Hence, the complexity of this search shows a polynomial increase with
the fineness of the grid. An important study would be to look into
change of decisions with slight change of transition probabilities.
%
%
%
%
%
\subsubsection{The case of multi-state arms}
The RMAB with partially observable states is generally analyzed in
literature for two state model. However, in many applications larger
number of states might be required for accurate representation of
communication channels. In such cases, the proposed model can be
applied using state aggregation. 
Further,
limited work is available on multi-state model for partially
observable RMAB, see \cite{Larranaga16}. For a multi-state single
armed bandit, the optimal policy of the corresponding POMDP is not
threshold type. Hence, in \cite{Larranaga16}, the POMDP is modified by
making steady state approximations; the modified POMDP has optimal policy of
threshold type. The approximations in \cite{Larranaga16} are valid for
large number of arms and large number of states. The applicability of
this analysis while allowing multiple state transitions per session,
needs careful study.
%

\end{document}